\documentclass[12pt,a4paper]{amsart}

\let\oldtocsection=\tocsection
\let\oldtocsubsection=\tocsubsection
\let\oldtocsubsubsection=\tocsubsubsection
\renewcommand{\tocsection}[2]{\hspace{0em}\oldtocsection{#1}{#2}}
\renewcommand{\tocsubsection}[2]{\hspace{1em}\oldtocsubsection{#1}{#2}}
\renewcommand{\tocsubsubsection}[2]{\hspace{2em}\oldtocsubsubsection{#1}{#2}}

\usepackage[left=2.5cm, top=3cm, bottom=3cm, right=2.5cm]{geometry}

\usepackage{cancel}

\usepackage{amsmath,amssymb,amscd,amsthm}

\usepackage[colorlinks=true
,urlcolor=blue
,anchorcolor=blue
,citecolor=blue
,filecolor=blue
,linkcolor=blue
,menucolor=blue
,pagecolor=blue
,linktocpage=true
,pdfproducer=medialab
]{hyperref}

\numberwithin{equation}{section}

\theoremstyle{plain}  
\newtheorem{theorem}{Theorem}[section]
\newtheorem{lemma}[theorem]{Lemma}

\newtheorem{proposition}[theorem]{Proposition}
\newtheorem{definition}[theorem]{Definition}
\theoremstyle{definition}  
\newtheorem{example}[theorem]{Example}

\newtheorem{remark}[theorem]{Remark}

\newcommand{\la}{\lambda}

\newcommand{\arrow}{\rightarrow}
\newcommand{\map}{\mapsto}

\newcommand{\diff}[2]{\frac{d#1}{d#2}}
\newcommand{\Diff}[3]{\left . \frac{d}{d#2}#1\right |_{#3}}

\newcommand{\alg}{\mathfrak{g}}

\newcommand{\aalg}{\mathfrak{a}}

\newcommand{\Alg}{\mathcal{A}}
\newcommand{\tAlg}{\tilde{\mathcal{A}}}

\newcommand{\F}{\mathcal{F}}

\newcommand{\vf}[1]{\mathfrak{X}(#1)}
\newcommand{\of}[1]{\Lambda^{1}(#1)}
\newcommand{\smf}[1]{\mathcal{C}^\infty(#1)}

\newcommand{\Si}{\mathbb{S}^1}
\newcommand{\Cm}{\mathbb{C}}
\newcommand{\CP}{\mathbb{CP}^1}
\newcommand{\Z}{\mathbb{Z}}

\newcommand{\Rm}{\mathbb{R}}
\newcommand{\Km}{\mathbb{K}}
\newcommand{\T}{T}

\newcommand{\pr}{\partial}
\newcommand{\me}{\geqslant}
\newcommand{\les}{\leqslant}
\newcommand{\bra}[1]{\! \left (#1\right )}
\newcommand{\brac}[1]{\! \left [#1\right ]}

\newcommand{\cc}[1]{(\!(#1)\!)}

\newcommand{\pobr}[1]{\left \{#1\right \}}
\newcommand{\ddual}[1]{\bigl \langle #1 \bigr \rangle}
\newcommand{\dual}[1]{\langle #1 \rangle}

\newcommand{\pd}[2]{\frac{\partial #1}{\partial #2}}

\DeclareMathOperator{\tr}{tr}
\DeclareMathOperator{\Tr}{Tr}
\DeclareMathOperator{\Lie}{Lie}

\DeclareMathOperator{\ad}{ad}
\DeclareMathOperator{\res}{res}

\DeclareMathOperator{\Der}{Der}
\DeclareMathOperator{\id}{id}

\DeclareMathOperator{\End}{End}
\DeclareMathOperator{\Direct}{D}
\DeclareMathOperator{\spn}{span}
\DeclareMathOperator{\im}{Im}

\newcommand{\Dir}{\Direct\!}

\renewcommand{\H}{\mathcal{H}}

\renewcommand{\tilde}[1]{\widetilde{#1}}

\newcommand{\dnabla}{\widetilde{\nabla}}

\newcommand{\eqq}[1]{\begin{align*} #1 \end{align*}}
\newcommand{\eq}[1]{\begin{align} #1 \end{align}}

\newcommand{\eps}{\varepsilon}

\usepackage{enumitem}
\setlist{nolistsep}

\begin{document}
%\abovedisplayskip2ex \belowdisplayskip2ex

%\markright{B\l a\.zej M. Szablikowski}
%\markboth{B\l a\.zej M. Szablikowski}{Classical $r$-matrix like approach to Frobenius manifolds and WDVV equations}

\title[Classical $r$-matrix like approach to Frobenius manifolds]{Classical $r$-matrix like approach to Frobenius manifolds, WDVV equations and flat metrics}

%\hyphenation{}

\author{B\l a\.zej M. Szablikowski}
\date{}

\address{Faculty of Physics, Adam Mickiewicz University,
Umultowska 85, 61-614 Pozna\'n, Poland}
\email{bszablik@amu.edu.pl}

\begin{abstract}
A general scheme for construction of flat pencils of contravariant metrics and Frobenius manifolds as well as related solutions to WDVV associativity equations is formulated. The advantage is taken from the Rota-Baxter identity
and some relation being counterpart of the modified Yang-Baxter identity from the classical $r$-matrix formalism.
The scheme for the construction of Frobenius manifolds is illustrated on the algebras of formal Laurent series and meromorphic functions on Riemann sphere. 
\end{abstract}

\maketitle

%\vspace{-1.5em}

{\small \parskip0.1ex \tableofcontents}
%\tableofcontents

\section{Introduction}

Classical $r$-matrix theory \cite{Sem1} can be very useful  in the construction of almost all known classes of integrable field systems (see \cite{Sem2,BS} and references therein). 
Frobenius manifolds are intrinsically connected with the bi-Hamiltonian systems of hydrodynamic type, for which the above formalism can be naturally applied (see for instance \cite{Li,Szab}). Therefore, the primary objective of this article is a formulation of a scheme analogous to the classical $r$-matrix theory, which could be applied for the construction of Frobenius manifolds.

The main idea exploited in the paper involves the use of the Rota-Baxter algebras for the construction of Frobenius algebras. Consider the new multiplication
\eqq{
	a\circ b := \ell(a)b + a\ell(b)
}
defined in some commutative associative algebra. If the endomorphism $\ell$ satisfies the Rota-Baxter identity \cite{Bax,Rot,Guo}:
\eqq{
\ell\bra{\ell(a) b} + \ell\bra{a \ell(b)} - \ell(a)\ell(b) = \kappa\, ab,
}
of weight $\kappa$, then the new multiplication is associative. If the algebra
can be equipped with some trace form, then the following metric is automatically invariant:
\eqq{
	\eta\bra{a,b} := \tr\bra{a\circ b}\quad\Rightarrow\quad \eta\bra{a\circ b,c} = \eta\bra{a,b\circ c}.
}
Hence, if the new multiplication is unital we can generate in this way, in principal, nontrivial Frobenius
algebras. We show in the article how this idea can be extended to Frobenius algebras, appearing in  the cotangent bundles, of certain Frobenius manifolds.

The theory of Frobenius manifolds \cite{Dub0,Dub1} is a coordinate-free formulation 
of the Witten-Dijkgraaf-Verlinde-Verlinde (WDVV) associativity equations appearing in the context 
of the $2$-dimensional topological quantum field theories (TFT) \cite{Witt,DVV}. Frobenius manifolds
appear not only in theoretical physics but also in seemingly unrelated areas of contemporary mathematics such as enumerative geometry and quantum cohomology \cite{Man}, singularity theory \cite{Her} and integrable systems
\cite{Dub1,Gue}, see also the expository articles \cite{Hit, Aud}.

In fact the structure of a Frobenius manifold is equivalent
to a pencil of flat metrics satisfying some homogeneity conditions \cite{Dub1,Dub2}.
These pencils generate hydrodynamic (Dubrovin-Novikov) Poisson tensors
yielding bi-Hamiltonian structure for the so-called principal hierarchy that can be associated to any Frobenius manifold.
Such integrable hierarchies together with their bi-Hamiltonian structures can be efficiently constructed using the classical $r$-matrix formalism (see for instance \cite{Szab}).

More standard approach to the construction of Frobenius manifolds relies 
on the Landau-Ginzburg models and Saito's theory \cite{Dub1}. 
It is interesting that often in these formalisms the superpotentials can be identified with Lax functions of the related Lax hierarchies. At the general level, it does not seem to be so apparent and the further research on the mutual connections
is justified. The above point of view for construction of Frobenius manifolds
is presented for instance in the following articles \cite{K2,AK,AK2,Str1,Str2,Car}, which are directly connected with 
integrable hierarchies of hydrodynamic type so-called Whitham hierarchies. Alternative approach based on the so-called isotropic deformations is presented in \cite{KM,Kon}.
For some recent works concerning classification of semisimple Frobenius manifolds see \cite{Din1,Dub3,Din2}.

A new interesting class of Frobenius manifolds, which are infinite-dimensional, was introduced  
in the recent articles \cite{CDM} and \cite{Rai}. These Frobenius manifolds are associated 
with $(2+1)$-dimensional integrable hydrodynamic 2d Toda and KP equations, respectively. 
Both works rely of the bi-Hamiltonian structures of the related infinite-field hydrodynamic chains. 
Nevertheless, both approaches are not equivalent as  they exploit in the construction of Frobenius 
structures significantly different mathematical methods. The approach from \cite{CDM} is further extended in \cite{WX} and \cite{WZ} to infinite-dimensional Frobenius manifolds associated with the two-component dispersionless BKP and Toda hierarchies defined on the appropriate space of pairs of meromorphic functions with possibly higher-order poles at the origin and at infinity.
 In principle, the approach presented in this article can be used in the above cases of infinite-dimensional Frobenius manifolds.

In Section~\ref{pre} we present all the facts about Frobenius manifolds and related structures that will be indispensable in the rest of the paper. In particular Proposition~\ref{Fpro} will allow for the straightforward construction of solutions to the WDVV equations. In Section~\ref{fro} we establish scheme for the construction
of Frobenius algebras, which is based on the Rota-Baxter identity \eqref{RB}. In Section~\ref{flat} 
we present scheme for the construction of flat metrics, which is based on the relation \eqref{rel}
being a counterpart of the modified Yang-Baxter identity from the classical $r$-matrix formalism.
In Section~\ref{Fman} the scheme for the construction of Frobenius manifolds is developed, which 
is based on the results from two preceding sections. Here, we prove the main Theorem~\ref{main} of this article.
In the last Section~\ref{mer} we illustrate our scheme of construction of Frobenius manifolds applying 
it to the algebras of formal Laurent series and meromorphic functions on Riemann sphere.
As a result, we generate infinite- and finite-dimensional Frobenius manifolds, respectively. 

We believe that the approach presented in this article will contribute to a better understanding
of the relation, on the constructive level, between Frobenius manifolds and integrable hydrodynamic
systems. This follows from the fact that the scheme for the construction of Frobenius manifolds
is formulated purely in the cotangent bundle, which is more natural in the context of the related
hydrodynamic bi-Hamiltonian structures. Moreover, we hope that this work will contribute
to the further classification of Frobenius manifolds, including particularly these which are infinite-dimensional.

\section{Theory of Frobenius manifolds}\label{pre}
 
\subsection{Frobenius manifolds} 

In this section we present all the necessary facts about Frobenius manifolds and related structures to make
the article self-contained. For the convention used see Appendix~\ref{A}.

\begin{definition}[\cite{Dub1}]\label{frob} 
A Frobenius manifold is an $n$-dimensional smooth manifold\footnote{All structures will be considered 
over field of real or complex numbers, that is $\Km=\Rm$ or $\Cm$.}
equipped with a (pseudo-Riemannian) covariant metric
$\eta\in\Gamma\bra{S^2\T^*M}$ 
and structure of a Frobenius algebra on the tangent bundle $\T M$. The last statement means 
that there exists an unital commutative associative multiplication,
given by the $\smf{M}$-bilinear map 
\eq{\label{mult}
*: \vf{M}\times\vf{M}\arrow\vf{M},
}
compatible with the metric:
\eq{\label{inv}
\eta(X,Y*Z)=\eta(X*Y,Z)\qquad X,Y,Z\in \vf{M}.
}
Let $\nabla$ be the Levi-Civita connection of $\eta$. The following conditions are also required:
\begin{enumerate}
\item The metric $\eta$ must be flat.
\item The tensor field $\nabla c$ must be symmetric in all its four arguments,
where $c(X,Y,Z):=\eta(X*Y,Z)$.
\item The unit vector field $e$ must be flat, that is $\nabla e =0$.
\item There must exists an Euler vector field $E$ (i.e. $\nabla\nabla E=0$) such that
\eq{\label{qh}
\Lie_E * = *\qquad \text{and}\qquad
\Lie_E \eta = (2-d)\, \eta, 
}
where $d$ is some number (weight).
\end{enumerate}
\end{definition}

The above structure without point (4) will be referred to as a pre-Frobenius manifold.

According to Theorem~2.15 in \cite{Her} the following conditions on any Frobenius manifold
are equivalent:
\begin{itemize}
\item[i)] The tensor $\nabla c$ is symmetric in all four arguments.
\item[ii)] The tensor $\nabla *$ is symmetric in all three arguments.
\item[iii)] The multiplication \eqref{mult} satisfies the $F$-manifold condition,
\eq{\label{f-man}
\Lie_{X*Y}(*) = X*\Lie_Y(*) + Y*\Lie_X(*),
}
and the counity 1-form $\eps := e^\flat$ is closed.
\end{itemize}
Moreover, from Lemma~2.16 in \cite{Her} we know that the vector field $e$ is flat if and only if 
\eq{\label{cond}
\Lie_e\eta = 0}
and $d\eps = 0$. Hence, in Definition~\ref{frob}  instead of the third condition we can require \eqref{cond} to hold. Direct consequence of \eqref{f-man} is that $\Lie_e (*) = 0$ and $\Lie_E e = -e$.

Substituting $e$ for $Z$ in \eqref{inv} one observes that the counity is actually 
a trace form, $\eps:\vf{M}\arrow\smf{M}$, such that
$\eta\bra{X,Y} = \eps\bra{X*Y}$.

\subsection{WDVV associativity equations}

Let $\{t^1,\ldots,t^n\}$ be (local) flat coordinates for the metric $\eta$ such that $e=\pr_{t^1}$.   
Then, the second condition in Definition~\ref{frob} implies local existence of the (smooth) function $\F = \F(t)$, the so-called prepotential, such that
\eq{\label{c}
c_{ijk} = \frac{\pr^3\F}{\pr t^i\pr t^j\pr t^k} \qquad \text{and}\qquad
\eta_{ij} = \frac{\pr^3\F}{\pr t^1\pr t^i\pr t^j}.
}
The structure constants for the multiplication \eqref{mult}, such that $(X*Y)^i = c^i_{jk}X^jY^k$, are given by $c^i_{jk} = \eta^{il}c_{ljk}$.\footnote{In this section the Einstein summation convention is used.}
Then, the so-called WDVV equations are the associativity equations on the prepotential $\F$: \eqq{
\frac{\pr^3\F}{\pr t^i\pr t^j\pr t^r}\eta^{rs}
\frac{\pr^3\F}{\pr t^s\pr t^k\pr t^l} = \frac{\pr^3\F}{\pr t^l\pr t^j\pr t^r}\eta^{rs} 
\frac{\pr^3\F}{\pr t^s\pr t^k\pr t^i}.
}
The Euler vector field $E$ can be normalised so that
$E^i = (1-q_i)t^i+r_i$, where $q_i$, $r_i$ are some constants. Furthermore, the quasi-homogeneity condition
\eqref{qh} on $\F$ takes the form  
\eq{\label{hom}
	\Lie_E \F \equiv E^i\frac{\pr \F}{\pr t^i} = (3-d)\F + quad.\ pol.
}
The above equality holds modulo quadratic polynomials in the flat coordinates.

\subsection{Intersection form}

On a Frobenius manifold the metric $\eta$  induces structure of a Frobenius algebra in the cotangent bundle $\T^*M$
with the multiplication given by
\eqq{
\alpha\circ\beta := \bra{\alpha^\sharp*\beta^\sharp}^\flat\qquad \alpha,\beta\in\of{M}.
} 
Its unity is the counity $1$-form $\eps = e^\flat$ and the unity vector field $e$ is a trace form such that  $\eta^*\bra{\alpha,\beta}= e\bra{\alpha\circ\beta}$. The quasi-homogeneity relations \eqref{qh} can be rewritten in the form:
\eq{\label{qhr}
\Lie_E \circ = (d-1)\circ,\qquad
\Lie_E \eta^* = (d-2)\, \eta^*.
}

Besides, on any Frobenius manifold there exists a second contravariant metric $g^*\in\Gamma\bra{S^2\T M}$, the so-called intersection form \cite{Dub1,Dub2}, defined by
\eq{\label{if}
g^*(\alpha,\beta) := \dual{\alpha\circ\beta, E}\qquad \alpha,\beta\in\of{M},
}
where $E$ is the Euler vector field. In fact, this metric is also flat and $g^*$ together with $\eta^*$ are compatible, 
that is the (contravariant) pencil defined by $g^{*}_{z} := \eta^* + z\, g^*$ is flat for all values of $z$. 

\begin{remark}
The existence of the intersection form $g^*$ is the main source of the
connection of Frobenius manifolds with integrable systems of hydrodynamic type. 
This is because the flat metrics $\eta^*$ and $g^*$ generate compatible Poisson brackets of hydrodynamic type, see Appendix~\ref{hydro}.  
\end{remark}

\subsection{Deformed flat connection}

On a Frobenius manifold one can define affine connection in the form
\begin{equation}\label{defc}
 \dnabla_X Y := \nabla_X Y + z\, X*Y\qquad X,Y\in\vf{M},
\end{equation}
where $z\in\Cm^*$ is a deformation parameter. This connection is torsionless (symmetric) and its curvature tensor vanish
identically in $z$. The symmetry of \eqref{defc} is equivalent to the 
commutativity of the multiplication \eqref{mult} and
its flatness is equivalent to the associativity of \eqref{mult} as well as symmetry of the tensor 
$\nabla c$ with respect to all its arguments. 

\begin{proposition}
The action of the deformed connection \eqref{defc} on $1$-forms is given by 
\begin{equation}\label{defc2}
\dnabla_{\alpha^\sharp} \gamma = \nabla_{\alpha^\sharp} \gamma - z\, \alpha\circ \gamma,
\end{equation}    
where $\alpha,\gamma\in\of{M}$.
\end{proposition}
\begin{proof}
By the invariance \eqref{inv} one finds that
\eqq{
\dual{\gamma,\alpha^\sharp*X} = \eta(\gamma^\sharp,\alpha^\sharp*X) =
\eta(\alpha^\sharp*\gamma^\sharp,X) = \dual{\alpha\circ\gamma,X}.
}
Hence, using the properties of affine connection we have
\eqq{
\dual{\dnabla_{\alpha^\sharp}\gamma,X} &= \Dir_{\alpha^\sharp} \dual{\gamma,X} - \dual{\gamma,\dnabla_{\alpha^\sharp}X} = \dual{\nabla_{\alpha^\sharp}\gamma,X} - z\, \dual{\gamma,\alpha^\sharp*X}\\
&= \dual{\nabla_{\alpha^\sharp} \gamma,X} - z\, \dual{\alpha\circ \gamma,X},
}
which gives the formula \eqref{defc}.
\end{proof}

Flatness of the deformed connection $\dnabla$ entails (local) existence of its flat coordinates 
$\H^k(z) \equiv \H^k(t, z)$ such that
\eqq{
	\dnabla_j d \H^k(z) = 0\qquad\iff\qquad 
	\frac{\pr^2 \H^k(z)}{\pr t^i\pr t^j} = z\, c_{ij}^l\, \frac{\pr \H^k(z)}{\pr t^l},
}
where $\nabla_j\equiv \nabla_{\frac{\pr}{\pr t^j}}$. 
One can expand $\H^k(z)$ into the formal power series 
\eqq{
\H^k(z) := \sum_{n=0}^\infty \H^k_{(n)}(t)z^n.
} 
Then, the coefficient functions $\H^k_{(n)}(t)$ can be determined recursively from
\eqq{
	\frac{\pr^2 \H^k_{(n)}}{\pr t^i\pr t^j}  = 
c_{ij}^l\, \frac{\pr \H^k_{(n-1)}}{\pr t^l} \qquad n>0.
} 
This recurrence formula has the following coordinate-free form: 
\eq{\label{rec2}
\nabla_{\alpha^\sharp}d\H^k_{(n)} = \alpha\circ d\H^k_{(n-1)}\qquad n>0,
}
valid for arbitrary $\alpha\in\of{M}$.

\begin{proposition}\label{Fpro}
Assuming normalization $\H^k_{(0)}= t^k$, if $d\neq 3$
the prepotential $\F$ can be determined from $\H^k_{(1)}$ using the formula
\eq{\label{Fpro1}
	\F = \frac{1}{3-d}\sum_{i,j} \eta_{ij}E^i \H^j_{(1)} + quad.\ pol.,
}
where the equality holds modulo quadratic polynomials in the flat coordinates.
\end{proposition}
\begin{proof}
Using the normalization and \eqref{c} we have
\eqq{
\frac{\pr^2 \H^k_{(1)}}{\pr t^i\pr t^j} = c_{ij}^l\frac{\pr \H^k_{(0)}}{\pr t^l} 
= c_{ij}^k = \eta^{kl} c_{ijl} = \eta^{kl}\frac{\pr^3 \F}{\pr t^i\pr t^j\pr t^l}.
}
Integrating twice, one finds
\eqq{
\H^k_{(1)} = \eta^{kl}\frac{\pr \F}{\pr t^l} + lin.\ pol.\quad\Rightarrow\quad 
\frac{\pr \F}{\pr t^i} = \eta_{ik} \H^k_{(1)} + lin.\ pol.\ ,
}
where the equalities hold modulo linear functions.
Next, substituting this to \eqref{hom} one obtains the desired formula \eqref{Fpro1}.
\end{proof}

\section{Construction of Frobenius algebras}\label{fro}

\subsection{Rota-Baxter identity}

Recall that a Frobenius algebra is an associative commutative unital algebra $\Alg$
endowed with a nondegenerate invariant symmetric bilinear scalar form.

\begin{proposition}\label{prop1}
Let $\Alg$ be an associative (non-necessarily commutative) algebra. 
\begin{itemize}
\item[i)] Define new multiplication on $\Alg$: 
\eq{\label{newm}
	a\circ_\ell b := \ell(a)b + a\ell(b)\qquad a,b\in\Alg,
}
generated by some linear map $\ell:\Alg\arrow\Alg$.
A sufficient condition for the multiplication \eqref{newm} to be associative
is the identity
\eq{\label{RB}
	\ell(a\circ_\ell b) - \ell(a)\ell(b) = \kappa\, ab, 	  	
}
which must hold for all $a,b\in\Alg$ and some $\kappa\in\mathrm{Z}(\Alg)$ (center of $\Alg$). 

\item[ii)] Moreover, the linear map $\tilde{\ell}(\cdot)=\ell(\delta\, \cdot)$, where $\ell$ is composed with  some $\delta\in\mathrm{Z}(\Alg)$, 
satisfies \eqref{RB}, with $\kappa$ replaced by $\tilde{\kappa} = \kappa\, \delta^2$, iff $\ell$ 
satisfies~\eqref{RB}.     
\end{itemize}
\end{proposition}
\begin{proof}
From the formula \eqref{RB},
\eqq{
&\bra{a\circ_\ell b}\circ_\ell c - a\circ_\ell\bra{b\circ_\ell c} =
  \ell\bra{a\circ_\ell b}c + \bra{a\circ_\ell b}\ell(c) - \ell(a)\bra{b\circ_\ell c}
- a\, \ell\bra{b\circ_\ell c}=\\
&\qquad\qquad = \brac{\ell(a\circ_\ell b) - \ell(a)\ell(b)}c -a \brac{\ell(b\circ_\ell c) - \ell(b)\ell(c)} = \kappa abc - a\kappa bc.
}
Thus, the first assertion follows immediately if $\kappa$ commutes with all elements from $\Alg$. For $\tilde{\ell}$
the right-hand side of \eqref{RB} takes the form
\eqq{
\tilde{\ell}(a\circ_{\tilde{\ell}}b)-\tilde{\ell}(a)\tilde{\ell}(b) = \ell\bra{(\delta a)\circ_\ell (\delta b)} - \ell(\delta a)\ell(\delta b)
= \kappa\delta^2 ab.
}
Hence, the second assertion follows.
\end{proof}

\begin{remark}
The formula \eqref{RB} is known as the Rota-Baxter identity \cite{Bax,Rot}. In most cases, when $\Alg$ is unital, $\kappa$  is a scalar weight. Associative algebras equipped with an operator satisfying the identity \eqref{RB}
are called Rota-Baxter algebras, for information on the subject see \cite{Guo} and references therein.
\end{remark}

\subsection{Invariant scalar product} 

Let $\Alg$ be a commutative associative unital algebra with a trace form given by a linear map $\tr:\Alg\arrow\Km$
such that the pairing 
\eqq{(\cdot,\cdot)_\Alg:\Alg\times\Alg\arrow \Km\qquad (a,b)_\Alg:= \tr(ab),
}
is nondegenerate. We will  call such trace nondegenerate.  

On the other hand, for an unital associative commutative algebra $\Alg$ equipped with a nondegenerate invariant pairing the trace form can be always defined as $\tr{a}: = (a,1)_\Alg$. This pairing is naturally invariant, hence such $\Alg$ is a Frobenius algebra. However, our aim is the construction 
of a {\it more complex} Frobenius structure on $\Alg$ with a scalar product invariant with respect to the commutative multiplication \eqref{newm}.

Let $\circ_\ell$ defined by \eqref{newm} be a second commutative associative multiplication on~$\Alg$. 
Then, we can define bilinear form (metric):
\eq{\label{nmet}
	\eta(a,b) := \tr\bra{a\circ_\ell b}\qquad a,b\in\Alg,	
} 
naturally invariant with respect to the multiplication $\circ_\ell$, that is
\eqq{
	\eta\bra{a\circ_\ell b, c} = \tr\bra{a\circ_\ell b\circ_\ell c} 
	= \eta\bra{a,b\circ_\ell c}.   	
}  
If the new multiplication \eqref{newm} is unital and such that \eqref{nmet} is nondegenerate, then 
the multiplication \eqref{newm} together with the metric \eqref{nmet} define structure of a Frobenius algebra on~$\Alg$. 

The relation from the following proposition will be needed later.

\begin{proposition} 
For the algebra $\Alg$ endowed with a nondegenerate inner product,
the Rota-Baxter identity \eqref{RB} is equivalent  to  the following 'dual'
relation:
\eq{\label{dRB}
\ell^*\bra{\ell^*(a)b}-\ell^*\bra{a\,\ell(b)}  +  \ell^*(a)\ell(b) = \kappa\, ab,
}
where $\ell^*$ is the adjoint of $\ell$ such that $\tr(\ell^*(a)b):=\tr(a\ell(b))$.
\end{proposition}
\begin{proof}
Define functionals in the form: 
\eqq{
K_1\brac{a,b}&:= \ell\bra{\ell(a) b} + \ell\bra{a\, \ell(b)} - \ell(a)\ell(b) - \kappa\ a b,\\
K_2\brac{a,b}&:=\ell^*\bra{\ell^*(a)b}-\ell^*\bra{a\,\ell(b)}  +  \ell^*(a)\ell(b) - \kappa\, ab,
}
which vanishing is equivalent to the identities \eqref{RB} and \eqref{dRB}.
The lemma follows from the equality
\eqq{
\tr\bra{K_1\brac{a,b}c} = \tr\bra{a K_2\brac{c,b}} 
}
and the fact that the inner product defined by the trace form is assumed to be nondegenerate.
\end{proof}

\subsection{Special case}

There is a class of simple solutions to the Rota-Baxter identity \eqref{RB} that will be of interest to us. 
Assume that the algebra $\Alg$ can be decomposed into a (vector) direct sum
of subalgebras preserving the multiplication, that is
\eqq{
	\Alg = \Alg_+\oplus\Alg_-\qquad \Alg_\pm\Alg_\pm\subset\Alg_\pm\qquad 
\Alg_+\cap\Alg_-= \pobr{0}  . 	
}   
Denoting the projections onto this subalgebras by $P_\pm$, we define $\ell:\Alg\arrow\Alg$ by
\eq{\label{pm}
	\ell = \frac{1}{2}\bra{P_+-P_-}.
}

\begin{proposition}\label{rbs}  
The linear map \eqref{pm} satisfies the identity \eqref{RB} for $\kappa =\frac{1}{4}$.
\end{proposition}
\begin{proof}
Set $a_\pm\equiv P_\pm(a)$ for $a\in\Alg$, thus $a=a_+ + a_-$. Then,
\eqq{
&\ell\bra{a\circ_\ell b} = \ell\bra{a_+b_+-a_-b_-} = \frac{1}{2}\bra{a_+b_++a_-b_-},\\ 
&\ell(a)\ell(b)= \frac{1}{4}\bra{a_+b_+-a_+b_--a_-b_++a_-b_-}.
}
Now, the result follows from a simple verification.
\end{proof}

\begin{remark}
The above special solution to the Rota-Baxter identity is a counterpart
of a similar construction for the modified Yang-Baxter equation, see Appendix~\ref{crm}.
\end{remark}

\section{Construction of flat metrics}\label{flat}

Let $\Alg$ be a commutative associative unital algebra equipped with a trace form $\tr:\Alg\arrow \Km$ such that the symmetric product $(a,b)_\Alg:= \tr\bra{ab}$  is nondegenerate. 
Furthermore, let be given some derivation $\pr\in\Der\Alg$ invariant with respect to the trace form, that is
\eq{\label{invd}
\tr\bra{a' b} = -\tr\bra{a b'}.
} 
Here and later we will frequently use the following notation for the above derivation:
\eq{\label{notation}
	a'\equiv \pr a.
}

Let $\Alg_M\subset  \Alg$ constitute a subspace (submanifold) of $\Alg$
corresponding to an underlying manifold $M$ embedded in $\Alg$, that is $\Alg_M$
is the image of the embedding. At each point $\la\in\Alg_M$ the tangent space $\T_\la\Alg_M$ 
can be uniquely identified with the vector subset of $\Alg$ defined by
\eqq{
	\T_\la\Alg_M \equiv \pobr{\left . \dot{\gamma} \right |_{t_0}\,|\ \gamma:I\arrow\Alg_M\ \text{s.t.}\ \gamma(t_0)=\la},
}
where $I$ is some interval containg $t_0$ and $\gamma$ means a smooth curve in $\Alg_M$.  Accordingly, each vector field $X\in\vf{\Alg_M}$ defines a smooth map $\Alg_M\arrow\Alg$ s.t.~$\la\map \left . X\right |_\la\in\T_\la\Alg_M$.

Using the nondegenerate symmetric product on $\Alg$
another vector subset can be identified as the cotangent space $\T_\la^*\Alg_M$, such that 
for each $\la\in\Alg_M$  the duality pairing takes the form
\eq{\label{dualp}
\dual{\ ,\ }_\la:\T_\la^*\Alg_M\times\T_\la\Alg_M\arrow \Km\qquad
\dual{\alpha,X}_\la:= (X,\alpha)_\Alg = \tr(X\alpha).
}  
Note that $\T_\la^*\Alg_M$ is defined modulo the orthogonal complement of $\T_\la\Alg_M$:
\eq{\label{oc}
	(\T_\la\Alg_M)^\perp \equiv \pobr{\alpha\in\Alg\, |\ \tr(X\alpha) = 0\ \text{for all}\ X\in\T_\la\Alg_M}.
}
In fact the cotangent space $\T_\la^*\Alg_M$ can be identified with the quotient space $\Alg/(\T_\la\Alg_M)^\perp$.
Choosing representatives of the cotangent spaces so that $\T_\la^*\Alg_M$ varies smoothly with respect to $\la\in\Alg_M$, we can also associate with each covector field (differential $1$-form) on $\Alg_M$, $\gamma\in\of{\Alg_M}$, a smooth map $\Alg_M\arrow\Alg$ s.t.~$\la\map \left . \gamma\right |_\la\in\T_\la^*\Alg_M$.

\begin{remark}\label{rde}
The notion of the directional (G\^ateaux) derivative, see Appendix~\ref{a1}, can be easily extended
onto (differentiable) maps $F:\Alg_M\arrow \Alg$ through the formula 
\eq{\label{edd}
		\bra{\Dir_X F}(\la) = \Diff{F\bra{\la+\varepsilon X}}{\varepsilon}{\varepsilon=0}\qquad X\in\vf{\Alg_M}.
}
Naturally the right-hand side must be computed within the algebra $\Alg$.
The differentiability here means that \eqref{edd} holds.  
In particular we have the equality $\Dir_X\la = X$ for arbitrary $\la\in\Alg_M$.
For instance the power function $F(\la) = \la^n$, where $n\in\mathbb{N}$ and $\la\in\Alg_M$,
defines a map $F:\Alg_M\arrow \Alg$ and in this case the derivative \eqref{edd} is simple to compute, that is $\Dir_XF = n\la^{n-1}X$.
In greater generality,  for power series functions $F=F(\la)$ of a single variable $\la\in\Alg_M$ the directional derivative \eqref{edd} exists and is given by $\Dir_XF = \frac{d F}{d \la}X$. Here the right-hand side is well-defined since both factors take values in the (commutative) algebra~$\Alg$. 
\end{remark}

\begin{remark}
We will say that the map $F:\Alg\arrow\Alg$ is invariant on $\Alg_M$ if the relation
\eq{\label{rr1}
	\Dir_X (F\circ G) = F\circ \Dir_X G
}
holds for arbitrary differentiable (in the sense of \eqref{edd}) map $G:\Alg_M\arrow \Alg$  and $X\in\vf{\Alg_M}$.
This means that the directional derivative commutes with invariant maps. In particular, composing an invariant 
map $F:\Alg\arrow\Alg$ with the inclusion $\iota:\Alg_M\hookrightarrow \Alg$, $\la\map \la$, 
we have the relation
\eq{\label{rr2}
	\brac{\Dir_X (F\circ \iota)}(\la)\equiv \Dir_X F(\la) = F(X)\qquad \la\in\Alg_M.
}
In \eqref{rr1} and \eqref{rr2}  $\circ$ means the composition of maps and it shall not be confused with the multiplication
denoted later by the same symbol.
\end{remark}

The another assumption we made is that the derivation $\pr$ be invariant on $\Alg_M$.
This means that $\pr$ commutes with the directional derivatives associated with $X\in\vf{\Alg_M}$ and in particular that 
$\Dir_X\la' = X'$, where $\la\in\Alg_M$.

\begin{remark}
From the all above assumptions it follows that the algebra $\Alg$, at least in principle, is infinite-dimensional and it has degrees of freedom related to the derivation $\pr$ and the underlying manifold $M$. A simple 'prototype' example of an algebra satisfying
the above restrictions is the algebra of formal Laurent series (at $\infty$) $\Cm\cc{p^{-1}}$\footnote{\label{foot}By $\Km\cc{x}$ we understand the algebra of formal infinite series in $x$ and $x^{-1}$ with coefficients from the field $\Km$ and only finitely many nonzero terms of negative degree.} with the derivation $\pr:=\pr_p$ and the the trace defined by means of the residue, that is $\tr(\sum_i a_ip^i) := a_{-1}$. For instance the (infinite-dimensional) submanifold $\Alg_M$ can be chosen as $\Alg_M
= \{p +\sum_{i\me 1}u_i p^{-i}\}$, where $\{u_i\}$ constitute the coordinates on $M$. The tangent space is given by $ \{\sum_{i\me 1}a_i p^{-i}\}$. One can easily check that $\pr$ is invariant on $\Alg_M$. For more details see Sections~\ref{ss1} and~\ref{ss3}. 
\end{remark}

Note that the algebra $\Alg$ constitute a framework for computations of structures that are defined on the underlying manifold $M$ associated with $\Alg_M$.

\subsection{Linear metric}

For some operator $r\in\End\Alg$ we define at each point $\la\in\Alg_M$
the following contravariant metric
\eq{\label{metric}
\eta_\la^*(\alpha, \beta) := \tr\bra{\la' r(\alpha)\beta +\la' \alpha\, r(\beta)},
}
where $\alpha,\beta\in\T_\la^*\Alg_M$. This metric \eqref{metric} will be formally called a linear metric due to the 
'explicit' first order dependence  on $\la$ in the formula \eqref{metric}.\footnote{Compare this with the terminology from Appendix~\ref{crm}.} The related canonical isomorphism 
$\left. \sharp \right |_\la:\T_\la^*\Alg_M\arrow\T_\la\Alg_M$,
such that 
\eq{\label{idcm0}
\eta^*(\alpha,\beta)  \equiv \tr\bra{\alpha^\sharp \beta},
}
is given by
\eq{\label{idcm}
\left . \alpha^\sharp \right|_\la  = \la' r(\alpha) + r^*(\la'\alpha),
}
where $r^*$ is the adjoint of $r$, i.e. $\tr\bra{r^*(a) b} := \tr\bra{a r(b)}$. 
For the nondegeneracy of the metric \eqref{metric} we require the kernel of $\sharp$
to be trivial at arbitrary $\la\in\Alg_M$. In practice, this requirement is possible to 
satisfy only outside some discriminant. Notice that from the identity \eqref{idcm0} it follows
that the metric \eqref{metric} naturally respects the quotient structure $\T_\la^*\Alg_M\cong \Alg/(\T_\la\Alg_M)^\perp$.

The following identity on the endomorphism $r$ turns out to be important:
 \eq{\label{rel}
r\bra{r(a) b'} + r\bra{a\, r(b)'} - r(a) r(b)'  = \kappa\, a b'\qquad a,b\in\Alg,
}
where $\kappa$ is some constant. 

\begin{theorem}\label{theorem1}
Assume that $r\in\End\Alg$ is invariant on $\Alg_M$, that is $r$ commutes with 
directional derivatives with respect to all vector fields $\vf{\Alg_M}$. Then, 
the following statements are valid:
\begin{itemize}
\item[(i)] If $r$ satisfies \eqref{rel}, the Levi-Civita connection of the metric \eqref{metric} has the form
\eq{\label{lcc}
\nabla_{\alpha^\sharp} \gamma = \Dir_{\alpha^\sharp}\gamma - r(\alpha)\gamma' -\alpha\, r(\gamma)'\qquad \alpha,\gamma\in\of{\Alg_M}.
}
\item[(ii)] The identity \eqref{rel} is a sufficient condition for the metric \eqref{metric} to be flat.
This means that  if $r$ satisfies \eqref{rel}  the curvature tensor vanishes on $\Alg_M$, that is
\eqq{
R(\alpha^\sharp,\beta^\sharp)\gamma\equiv \nabla_{\alpha^\sharp}\nabla_{\beta^\sharp}\gamma - \nabla_{\beta^\sharp}\nabla_{\alpha^\sharp}\gamma - \nabla_{[\alpha^\sharp,\beta^\sharp]}\gamma = 0,
}
where $\alpha,\beta,\gamma\in\of{\Alg_M}$.
\end{itemize}
\end{theorem}
\begin{proof}
The tensor field \eqref{gamma} corresponding to the connection \eqref{lcc} has the form
\eqq{
\Gamma_\alpha\gamma = \Dir_{\alpha^\sharp}\gamma - \nabla_{\alpha^\sharp}\gamma =
r(\alpha)\gamma' +\alpha\, r(\gamma)'.
}
We must to show that this is the Levi-Civita connection of the metric \eqref{metric},
that is the conditions \eqref{c1} and \eqref{c2} are satisfied. 

 The identity \eqref{rel} can be written in the form:  
\eq{\label{r}
r(\Gamma_a b) = r(a) r(b)' + \kappa\, a b'.
}

Let $a.b:= r(a)b + a\, r(b)$, so that $\eta^*(\alpha, \beta) = \tr (\la'\alpha.\beta)$.
Straightforward  computation, using \eqref{r}, leads to the following two relations:
\eq{\label{r0}
a.\Gamma_b c = b.\Gamma_a c,\qquad
\bra{a.b}' = \Gamma_a b + \Gamma_b a.
}
Now, the condition \eqref{c1} is immediate as
\eq{\label{r1}
\eta^*\bra{\alpha,\Gamma_\beta \gamma} &= \tr \bra{\la' \alpha.\Gamma_\beta\gamma} =
\tr \bra{\la' \beta.\Gamma_\alpha\gamma} = \eta^*\bra{\Gamma_\alpha\gamma,\beta}.
}

From the requirement that $\pr$ and $r$ commute with directional derivatives associated with vector fields on $\Alg_M$ it follows that $\Dir_{\alpha^\sharp}\la' = (\alpha^\sharp)'$
and that $\Gamma$ is constant on $\Alg_M$, that is
 $(\Dir_{\alpha^\sharp}\Gamma)(\beta,\gamma) = 0$.
Hence, the directional derivative of the metric \eqref{metric}, by \eqref{invd} and \eqref{r0}, is
\eqq{
\bra{\Dir_{\alpha^\sharp}\eta^*}(\beta,\gamma) &= \tr\bra{(\alpha^\sharp)'\beta.\gamma} 
= - \tr\bra{\alpha^\sharp (\beta.\gamma)'}
= -\eta^*\bra{\alpha,(\beta.\gamma)'}\\
&= -\eta^*\bra{\alpha,\Gamma_\beta\gamma} - \eta^*\bra{\alpha,\Gamma_\gamma\beta}. 
}
Using \eqref{r1} we get the second condition \eqref{c2}:
\eqq{
\bra{\Dir_{\alpha^\sharp}\eta^*}(\beta,\gamma) 
= -\eta^*\bra{\beta,\Gamma_\alpha\gamma} -\eta^*\bra{\Gamma_\alpha\beta,\gamma}.
}
Hence,  indeed the formula \eqref{lcc} defines the Levi-Civita connection. 

Since $\Gamma$ is invariant on $\Alg_M$ the curvature tensor \eqref{ct} for the metric \eqref{metric}
takes the form
\eqq{
R(\alpha^\sharp,\beta^\sharp)\gamma =  
\Gamma_\alpha\Gamma_\beta\gamma - \Gamma_\beta\Gamma_\alpha\gamma 
- \Gamma_{\Gamma_\alpha\beta}\gamma + \Gamma_{\Gamma_\beta\alpha}\gamma = 0,
}
where the last equality is a consequence of \eqref{r} and straightforward computation.
Thus, the metric is flat.
\end{proof}

\subsection{General case} \label{subs}

We define the generalised contravariant metric, for $r\in\End\Alg$, by
\eq{\label{gmet}
g_\la^*(\alpha, \beta) &:= \tr\bra{\la' r(E \alpha)\beta +\la' \alpha\, r(E \beta)}\\\notag
&\ \equiv \tr\bra{E\brac{r^*(\la' \alpha)\beta + \alpha\, r^*(\la' \beta)}},
}
where $\la\in\Alg_M$ and $\alpha,\beta\in\T_\la^*\Alg_M$.  We assume that the map $E:\Alg_M\arrow \Alg$ is a differentiable
function $E=E(\la)$ as a function of a single variable $\la$, such that 
$\Dir_{X}E = \frac{d E}{d \la} X$ holds for arbitrary $X\in\vf{\Alg_M}$.\footnote{See Remark~\ref{rde}.} Then, also $E' = \frac{d E}{d \la} \la'$.\footnote{Recall that we use the notation \eqref{notation}.} For instance one could take $E=\la^n$ or more complicated function defined by means of power series.

The related canonical isomorphism 
$\left. \sharp \right |_\la:\T_\la^*\Alg_M\arrow\T_\la\Alg_M$,
such that $g^*(\alpha,\beta)  \equiv \tr\bra{\alpha^\sharp \beta}$,
has the form
\eq{\label{cis}
\left . \alpha^\sharp \right|_\la  = \la' r(E \alpha) + E\, r^*(\la'\alpha).
}
We require the kernel of $\sharp$ to be trivial on $\Alg_M$.

\begin{theorem}\label{thmg}
Assume that $r\in\End\Alg$ satisfies \eqref{rel} and it is invariant on $\Alg_M$. Then, the following statements
are valid:
\begin{itemize}
\item[(i)] The Levi-Civita connection for the metric \eqref{gmet} is given by
\eq{\label{glcc}
\nabla_{\alpha^\sharp} \gamma = \Dir_{\alpha^\sharp}\gamma - r(E\alpha)\gamma' -\alpha\, r(E\gamma)'
 	+ \frac{d E}{d \la}\, r^*(\la'\alpha)\gamma ,
}
where $ \alpha,\gamma\in\of{\Alg_M}$.
\item[(ii)] The metric \eqref{gmet} is flat, that is the curvature tensor \eqref{ct} vanish identically on~$\Alg_M$.
\end{itemize}
\end{theorem}

\begin{lemma} 
The identity \eqref{rel}, for $r\in\End\Alg$, is equivalent  to  
\eq{\label{rela}
r^*\bra{r^*(a)b'}-r^*(a\,r(b)')  +  r^*(a)r(b)' =  \kappa\, ab'.
}
\end{lemma}
\begin{proof}
Define functionals in the form: 
\eq{\label{kk}
K_1\brac{a,b}&:= r\bra{r(a) b'} + r\bra{a\, r(b)'} - r(a) r(b)' - \kappa\ a b',\\\notag
K_2\brac{a,b}&:= r^*\bra{r^*(a)b'}-r^*(a\,r(b)')  +  r^*(a)r(b)' -  \kappa\, ab'.
}
Vanishing of these functionals is equivalent to the identities \eqref{rel} and \eqref{rela}, respectively.
The lemma follows from the equality
\eqq{
\tr\bra{K_1\brac{a,b}c} = \tr\bra{a K_2\brac{c,b}} 
}
and the fact that the trace form is assumed to be nondegenerate. 
\end{proof}

\begin{proof}[Proof of Theorem~\ref{thmg}]
The tensor field \eqref{gamma} related to \eqref{glcc} has the form
\eq{\label{g0}
\Gamma_\alpha\gamma = r(E\alpha)\gamma' +\alpha\, r(E\gamma)' - \frac{d E}{d \la} r^*(\la'\alpha)\gamma.
}
We must show that it satisfies \eqref{c1} and \eqref{c2}. Let
\eqq{
\alpha \diamond \beta := r(E\alpha)\beta + \alpha\, r(E \beta)\quad \Longrightarrow\quad g^*(\alpha,\beta) = \tr\bra{\la'\alpha\diamond\beta}	.
}
Using the relation \eqref{rel} and properties of the trace form one can see that 
\eq{\label{sss}
\tr \bra{\la' \alpha\diamond\Gamma_\beta\gamma} =
\tr \bra{\la' \beta\diamond\Gamma_\alpha\gamma},
}
which  is equivalent to \eqref{c1}. 

Computing the directional derivative of the metric \eqref{gmet} one finds the formula
\eqq{
		(\Dir_{\alpha^\sharp} g^*)(\beta,\gamma) &= 
 -\tr\bra{\alpha^\sharp\bra{r(E\beta)\gamma + \beta r(E \gamma)}'} 
+  \tr\Bigl (\alpha^\sharp\frac{d E}{d \la}\bra{\beta r^*(\la'\gamma)+ r(\la'\beta)\gamma}\Bigr )\\
&= -g^*\bra{\alpha,\Gamma_\beta\gamma} - g^*\bra{\alpha,\Gamma_\gamma\beta}.
}
Now, the second condition \eqref{c2} is a straightforward consequence of \eqref{c1} and \eqref{sss}. 

To calculate the curvature tensor \eqref{ct} first we need the directional derivative of  \eqref{g0}, which is
given by
\eqq{
	\bra{\Dir_{\beta^\sharp}\Gamma}\bra{\alpha,\gamma} = 
r\Bigl (\frac{d E}{d \la}\alpha\beta^\sharp\Bigr )\gamma' + \alpha\, r\Bigl (\frac{d E}{d \la}\beta^\sharp\gamma\Bigr )'
- \frac{d^2E}{d\la^2}r^*(\la'\alpha)\beta^\sharp\gamma 
- \frac{d E}{d\la}r^*\bra{\alpha(\beta^\sharp)'}\gamma,
}
where $\beta^\sharp  = \la' r(E \beta) + E\, r^*(\la'\beta)\in\vf{\Alg_M}$.
Substituting the above formula and \eqref{g0} to \eqref{ct} one can show that the curvature 
vanishes. One must use the identities \eqref{rel} and \eqref{rela}. The calculation is straightforward, however 
slightly tedious, so we omit it.
\end{proof}

\begin{remark}
Note that the generalized metric \eqref{gmet} and the $\Gamma$ tensors \eqref{g0}
are explicitly  linear in $E$. Hence, any two contravariant metrics \eqref{gmet} defined by different $E$
are compatible and their linear composition generates the corresponding contravariant flat pencil.  
Setting $E=1$ in \eqref{gmet}, where $1$ is the unity of the algebra $\Alg$, we obtain the linear metric \eqref{metric}. 
\end{remark}

\begin{proposition}
For arbitrary $X\in\vf{\Alg_M}$ and $\gamma\in\of{\Alg_M}$ the following formula holds:
\eq{\label{pf}
\nabla_X\gamma^\sharp\equiv \bra{\nabla_X\gamma}^\sharp = \pobr{\la, r(E\gamma)} + E\, r^*\bra{\pobr{\la,\gamma}},
}
where $\nabla$ is the Levi-Civita connection \eqref{glcc} for the metric \eqref{gmet}, $\sharp$
is the canonical isomorphism \eqref{cis} and
\eq{\label{pobr}
	\pobr{a,b}:= a' \Dir_Xb  - \Dir_Xa\, b'.
}
\end{proposition}
\begin{proof}
Applying the canonical isomorphism \eqref{cis} to \eqref{g0} and using the relations \eqref{rel} and \eqref{rela} 
one finds that
\eqq{
\bra{\Gamma_\alpha\gamma}^\sharp = \alpha^\sharp r(E\gamma)' + E r^*( \alpha^\sharp\gamma')
- \la' r\Bigl ( \frac{d E}{d \la}\alpha^\sharp \gamma\Bigr ),
}
where $\alpha^\sharp  = \la' r(E \alpha) + E\, r^*(\la'\alpha)$.
Hence, 
\eqq{
\bra{\nabla_{\alpha^\sharp}\gamma}^\sharp &= \bra{\Dir_{\alpha^\sharp}\gamma - \Gamma_\alpha\gamma}^\sharp\\
&= \la'\Dir_{\alpha^\sharp}r(E\gamma) - \alpha^\sharp r(E\gamma)' - \la'r^*(\la'\Dir_{\alpha^\sharp}\gamma - \alpha^\sharp\gamma').
}
Now, setting $X\equiv\alpha^\sharp\in\vf{\Alg_M}$ and 
defining $\pobr{,}:=\pr\wedge\Dir_X$ we get the formula \eqref{pf}.
Notice that $\Dir_X\la = X$. 
\end{proof}

The formula \eqref{pobr} has the form of a Poisson bracket and, in fact, for fixed $X\in\vf{\Alg_M}$ it defines a structure of a Poisson algebra in the space of differentiable maps $\Alg_M\arrow\Alg$. In this case the Poisson bracket is well-defined since it is assumed that the derivation~$\pr$ commutes with the directional derivatives with respect to vector fields on $\Alg_M$. 

\begin{proposition}
The condition \eqref{rel} is a sufficient condition for $r\in\End{\Alg}$, which is invariant on $\Alg_M$, to be a classical $r$-matrix
with respect to the Poisson bracket  \eqref{pobr} (see Appendix~\ref{crm}), which means that
\eqq{
r\bra{\pobr{r(f),g}} + r\bra{\pobr{f,r(g)}} - \pobr{r(f),r(g)} = \kappa \pobr{f,g}
}
holds. Here $f,g:\Alg_M\arrow\Alg$ are differentiable maps. 
\end{proposition}

The proof is straightforward using the relation \eqref{rel} and the invariance of $r$.

\begin{remark}\label{rem}
The formula \eqref{pf} can be useful in finding flat coordinates of the corresponding metric \eqref{gmet}.
To find them we can look for linearly independent flat (covariantly constant) $1$-forms~$\gamma^i$,
that is $\nabla \gamma^i = 0$. Using the relation \eqref{pf} it is sufficient to postulate
that
\eq{\label{flatc}
\pobr{\gamma^i,\la} = 0\qquad\text{and}\qquad \pobr{r(E\gamma^i),\la} = 0.
} 
Then, locally by Poincar\'e Lemma $\gamma^i = dt^i$ and all $t^i$ constitute flat coordinates.
Particularly, in this way we can obtain, taking $E= 1$, flat coordinates for the metric \eqref{metric}.
\end{remark}

\begin{remark}
For $E=1$, where $1$ is the unity of the algebra $\Alg$, the metric \eqref{gmet} reduces to the linear metric \eqref{metric}.
The most natural choice is $E=\la^n$ for $n\me 0$. (One can also imagine more 
nonstandard choices of $E$.)
The case $E=\la^n$ through the two above propositions, corresponds to the Lie-Poisson 
brackets from Theorem~\ref{liep}. Compare the formula \eqref{pf} with \eqref{ptn}. 
Under appropriate assumptions, the above formalism
gives alternative proof to Theorem~\ref{liep} and it can also be considered as its generalization. 
\end{remark}

\begin{remark} 
Consider a loop algebra $\mathcal{L}(\Alg):=\{\gamma:\Si\arrow \Alg\,\}$ 
consisting of loops in $\Alg$ such that the derivation $\pr$ is invariant along these loops.
This means that the tangent vector fields to loops must commute with $\pr$.
Then, taking $X\equiv \pr_x$, where $x\in\Si$, the formula \eqref{pobr}  defines a Poisson bracket on $\mathcal{L}(\Alg)$
and the classical $r$-matrix scheme could be applied to $\mathcal{L}(\Alg)$, see Appendix~\ref{crm}. This provides the close connection between the scheme for the construction of covariant metrics
presented in this section and the $r$-matrix approach for construction of Poisson algebras from Appendix~\ref{crm}.
\end{remark}

\subsection{Frobenius structure}

The linear metric \eqref{metric} can be written in the form
\eqq{
\eta_\la^*(\alpha, \beta) = \tr\bra{r^*(\la'\alpha)\beta + \alpha\, r^*(\la'\beta)},
}
which suggests that there could be defined in the cotangent bundle $\T^*\Alg_M$ invariant multiplication
acquiring  the form
\eq{\label{rmult}
		\alpha\circ\beta = r^*(\la'\alpha)\beta + \alpha\, r^*(\la'\beta)
}
such that $\eta^*(\alpha, \beta) = \tr\bra{\alpha\circ\beta}$. 

According to Proposition~\ref{prop1},  if $r^*$ satisfies the Rota-Baxter identity \eqref{RB}, then this multiplication is associative. In this case the related co-unity $1$-form would be
\eq{\label{co}
	\eps\bra{\alpha^\sharp} \equiv \dual{\eps,\alpha^\sharp} = \tr{\alpha}.
}
As we have seen in Section~\ref{pre}, on a Frobenius manifold the counity is necessarily closed. 

\begin{proposition}\label{prop}
Assume that $r\in\End{\Alg}$ satisfies \eqref{rel} so that Theorem~\ref{theorem1} holds. Then, the condition 
\eq{\label{sc}
r^*(\gamma') + r(\gamma)' = 0,
} 
where $\gamma\in\of{\Alg_M}$ is arbitrary, is a sufficient condition for vanishing of $d\eps$, where $\eps$ is the 
$1$-form defined by \eqref{co}.
\end{proposition}
\begin{proof}
The exterior derivative of a $1$-form $\eps$ is
\eqq{
d\eps\bra{X,Y} = X\bra{\eps(Y)} - Y\bra{\eps(X)} - \eps \bra{ \brac{X,Y}},
}
where $X,Y$ are vector fields on $\Alg_M$. Setting $X=\alpha^\sharp$ and $Y=\beta^\sharp$ and using \eqref{lieb},
it follows that for the $1$-form \eqref{co}:
\eqq{
d\eps(\alpha^\sharp,\beta^\sharp) &= 
\tr\bra{\Dir_{\alpha^\sharp}\beta - \Dir_{\beta^\sharp}\alpha -\nabla_{\alpha^\sharp}\beta + \nabla_{\beta^\sharp}\alpha}
= \tr\bra{\Gamma_\alpha\beta - \Gamma_\beta\alpha}\\
&= 2\tr\bra{\alpha \bra{r^*(\beta') + r(\beta)'}}.
}
Hence, the assertion follows.
\end{proof}

\begin{remark}
In fact, one can show that if $r$ satisfies the relation \eqref{rel},
the condition \eqref{sc}, and the derivation $\pr$ is {\it onto} ($\im\pr = \Alg$), then $r^*$
satisfies the Rota-Baxter identity \eqref{RB} and the multiplication \eqref{rmult} is associative. 
There arise question when one can define on $\Alg_M$ structure of a Frobenius
manifold using some endomorphism that satisfies the Rota-Baxter identity \eqref{RB}.
\end{remark}

\section{Construction of pre-Frobenius manifolds}\label{Fman}

The preliminary setting is the same as in the previous section. Let $\Alg$ be a commutative associative unital algebra equipped with a nondegenerate trace form $\tr:\Alg\arrow \Km$
and invariant derivation $\pr\in\Der\Alg$, that is $\tr {a'} = 0$, 
where $a'\equiv \pr a$. Consider the subspace $\Alg_M\subset  \Alg$, induced by an underlying manifold. Then, the tangent and cotangent spaces can be identified with appropriate vector subspaces of $\Alg$ through the trace form.
We also assume that the derivation $\pr$ is invariant on $\Alg_M$, that is it commutes with
all directional derivatives with respect to vector fields on $\Alg_M$.

Let be given some linear map $\ell\in\End\Alg$, which satisfies the Rota-Baxter identity \eqref{RB}, that is 
\eq{\label{RBi}
\ell\bra{\ell(a) b} + \ell\bra{a \ell(b)} - \ell(a)\ell(b) = \kappa\, ab,
}
for some $\kappa\in\Km$. Then, for fixed $\la\in\Alg$ we define the second commutative multiplication in $\Alg$ by the formula
\eq{\label{fmult}
a\circ b := \ell(\la' a)b + a\ell(\la'b).
}
By Proposition~\ref{prop1} this multiplication is associative.

\subsection{Structure of Frobenius algebra}

We define the following contravariant metric at a point  $\la\in\Alg_M$ by the formula
\eq{\label{fmetric}
\eta_\la^*(\alpha, \beta) := \tr\bra{\alpha\circ\beta}\equiv
\tr\bra{\ell(\la'\alpha)\beta + \alpha\, \ell(\la'\beta)}\qquad \alpha,\beta\in\T_\la^*\Alg_M
}
and require that on $\Alg_M$ the metric is nondegenerate. The related canonical isomorphism  
$\left. \sharp \right |_\la:\T_\la^*\Alg_M\arrow\T_\la\Alg_M$ is
\eq{\label{iso}
\left. \alpha^\sharp \right |_\la  = \ell(\la'\alpha) + \la' \ell^*(\alpha),
}
where $\ell^*$ is the adjoint of $\ell$ with respect to the trace form. 

If we suppose that the multiplication \eqref{fmult} restricts properly to $\T_\la^*\Alg_M$, then the formula \eqref{fmult}
defines  in the cotangent bundle associative and commutative multiplication, such that
\eq{\label{fm2}
\circ: \T_\la^*\Alg_M\times \T_\la^*\Alg_M\arrow  \T_\la^*\Alg_M\qquad
(\alpha,\beta)\map \circ (\alpha,\beta)\equiv \alpha\circ\beta.
} 
This multiplication is invariant with respect to the metric \eqref{fmetric}. Hence, the contravariant metric \eqref{fmetric} and the multiplication \eqref{fmult}, if it is unital, define the structure of a Frobenius algebra in the cotangent bundle $\T^*\Alg_M$.  

In practice the multiplication \eqref{fmult} does not have to naturally restrict  to $\T_\la^*\Alg_M$. Let us remind that we have the identification $\T_\la^*\Alg_M\cong \Alg/(\T_\la\Alg_M)^\perp$, where $(\T_\la\Alg_M)^\perp$ is the orthogonal complement \eqref{oc}. 
If $(\T_\la\Alg_M)^\perp$ is an ideal in $\Alg$ with respect to the multiplication \eqref{fmult}, then we can define multiplication in $\T_\la^*\Alg_M$ by means of the quotient algebra $\Alg/(\T_\la\Alg_M)^\perp$.
That is, we must require that 
\eq{\label{rq}
(\T_\la\Alg_M)^\perp\circ\Alg\subset (\T_\la\Alg_M)^\perp.
}
Still, we need the metric \eqref{fmetric} to be compatible with the quotient structure. 
Let us extend the formula \eqref{fmetric} to the whole algebra $\Alg$, that is  
\eqref{fmetric} for $\alpha,\beta\in\Alg$ defines bilinear form on $\Alg$.
The compatibility of the bilinear form \eqref{fmetric} with the quotient structure $\Alg/(\T_\la\Alg_M)^\perp$ imposes the following necessary condition
\eq{\label{req2}
	\eta_\la^*\bra{(\T_\la\Alg_M)^\perp,\Alg} = \tr\bra{(\T_\la\Alg_M)^\perp\circ\Alg} = 0.
}
In this case the metric \eqref{fmetric} is invariant with respect to the quotient structure $\Alg/(\T_\la\Alg_M)^\perp$.
Hence, we have the following proposition.

\begin{proposition}
Let the orthogonal complement $(\T_\la\Alg_M)^\perp$, for each $\la\in\Alg_M$,
be an ideal in the algebra $\Alg$ with respect to the multiplication \eqref{fmult}
such that \eqref{req2} holds.
Then, the multiplication \eqref{fmult} and the metric \eqref{fmetric} define the structure of a (nonunital) Frobenius algebra in the cotangent bundle $\T^*\Alg_M$.
\end{proposition}

\begin{remark}
Assume that the multiplication \eqref{fm2} is unital and the unity $1$-form is given by $\eps$. 
Then,  for arbitrary $\alpha\in\of{\Alg_M}$ we have $\dual{\eps,\alpha^\sharp} = \tr(\eps\circ\alpha)= \tr \alpha$.
On the other hand $\dual{\eps,\alpha^\sharp} = \dual{\alpha,e} \equiv \tr(e\alpha)$,
where $e=\eps^\sharp$ is the unity vector field. Since the trace is nondegenerate, we see that for the multiplication \eqref{fm2} the unity vector field $e\simeq 1$, that is $e$ coincides with the unity $1$ (modulo orthogonal complement of $\T_\la^*\Alg_M$) of the original algebra $\Alg$ or lies in the same equivalence class.  
\end{remark}

\subsection{Main theorem} The metric \eqref{fmetric} can be written in the form
\eqq{
\eta^*(\alpha, \beta) := \tr\bra{\la' \ell^*(\alpha)\beta +\la' \alpha \ell^*(\beta)},
}
which, when $r = \ell^*$, coincides with \eqref{metric}. By Theorem~\ref{theorem1}
the sufficient  condition for flatness of the metric is the identity \eqref{rel}.     
It turns out, that for $r = \ell^*$ the condition \eqref{rel} is fulfilled, if the Rota-Baxter identity \eqref{RBi} together with \eqref{sc} hold.   

\begin{lemma}\label{lemma1}
If $\ell\in\End{\Alg}$ satisfies the Rot-Baxter identity \eqref{RBi}
and the relation
\eq{\label{frel}
\ell(a') +\ell^*(a)' = 0, 
}
then $r = \ell^*$ fulfils the identity \eqref{rel}, that is 
\eqq{
\ell^*\bra{\ell^*(a) b'} + \ell^*\bra{a\ell^*(b)'} - \ell^*(a) \ell^*(b)' = \kappa\, a b',
}
and also
\eq{\label{rel3}
	\ell\bra{\ell^*(a) b'} - \ell\bra{a\ell(b)'} - \ell^*(a) \ell(b)' = \kappa\, a b'.	
}
\end{lemma}
\begin{proof}
Let
\eqq{
\tilde{K}_1\brac{a,b}&:=\ell\bra{\ell(a) b} + \ell\bra{a \ell(b)} - \ell(a)\ell(b) - \kappa\, ab,\\
\tilde{K}_2\brac{a,b}&:= \ell\bra{\ell^*(a) b'} - \ell\bra{a\ell(b)'} - \ell^*(a) \ell(b)' + \kappa\, a b',
}
which are connected with conditions \eqref{RBi} and \eqref{rel3}, respectively. Then,
\eqq{
\tr\bigl (\tilde{K}_1\brac{a,b'}c\bigr ) &= \tr\bigl (a K_1\brac{c,b}\bigr ),\\
\tr\bigl (\tilde{K}_2\brac{a,b}c\bigr ) &= \tr\bigl (b K_1\brac{a,c}-b K_1\brac{c,a}\bigr ), 
}
where $K_1$ is given by \eqref{kk} for $r = \ell^*$. Now, the results of the lemma
follows from the nondegeneracy of the trace form.
\end{proof}

We will show that under certain technical assumption on a submanifold $\Alg_M$ of $\Alg$
we can define the structure of a pre-Frobenius manifold.

\begin{theorem}\label{main}
Assume that $\ell\in\End{\Alg}$ is invariant on $\Alg_M$, that is it must commute with 
all directional derivatives with respect to vector fields $\vf{\Alg_M}$.
Let the endomorphism~$\ell$ satisfy the Rota-Baxter identity~\eqref{RBi}
and the requirement \eqref{frel}. 
Then, the following statements hold:
\begin{itemize}
\item[(i)] The Levi-Civita connection for the contravariant metric \eqref{fmetric} has the form
\eq{\label{lcc2}
\nabla_{\alpha^\sharp} \gamma = \Dir_{\alpha^\sharp}\gamma 
+ \alpha\, \ell(\gamma') - \ell^*(\alpha)\gamma'.
}
\item[(ii)] The metric \eqref{fmetric} is flat, that is the curvature tensor $R(\alpha^\sharp,\beta^\sharp)\gamma$ vanishes identically on $\Alg_M$.
\item[(iii)] The co-unity $1$-form $\eps$, such that $\eps(\alpha^\sharp) \equiv \tr\alpha$, is closed. 
\item[(iv)] The tensor $\nabla*$ is symmetric in all three arguments, where
\eqq{
		\alpha^\sharp*\beta^\sharp := (\alpha\circ\beta)^\sharp\qquad \alpha,\beta\in\of{\Alg_M}
}
is the induced multiplication in the tangent bundle $\T\Alg_M$. In principle, the relation 
\eq{\label{nsymm}
\bra{\nabla_{\alpha^\sharp}\circ}(\beta,\gamma) = \bra{\nabla_{\beta^\sharp}\circ}(\alpha,\gamma)
}
is valid.
\end{itemize}
\end{theorem}
\begin{proof}
Assuming that $r=\ell^*$, the first three points of the theorem are straightforward consequence 
of Lemma~\ref{lemma1}, Theorem~\ref{theorem1} and Proposition~\ref{prop}.

It is left to show that $\nabla*$  is symmetric in all its arguments, i.e.
\eqq{
\bra{\nabla_{\alpha^\sharp}*}(\beta^\sharp,\gamma^\sharp)
= \bra{\nabla_{\beta^\sharp}*}(\alpha^\sharp,\gamma^\sharp).
}
Since $\nabla$ is the Levi-Civita connection the following relation is valid:
\eqq{
\bra{\nabla_{\alpha^\sharp}*}(\beta^\sharp,\gamma^\sharp) 
= \bra{\bra{\nabla_{\alpha^\sharp}\circ}(\beta,\gamma)}^\sharp.
}
Hence, it is sufficient to show that \eqref{nsymm} holds.  
 
Expanding $\nabla\circ$, one finds that
\eq{\label{cdm}\begin{split}
\bra{\nabla_{\alpha^\sharp}\circ}(\beta,\gamma) &=
\nabla_{\alpha^\sharp}\bra{\beta\circ\gamma} - \nabla_{\alpha^\sharp}\beta\circ\gamma
- \beta\circ\nabla_{\alpha^\sharp}\gamma\\
&= (\Dir_{\alpha^\sharp}\circ)(\beta,\gamma) - \Gamma_\alpha(\beta\circ\gamma) + \Gamma_\alpha\beta\circ\gamma
+ \beta\circ\Gamma_\alpha\gamma,
\end{split}
}
where $\Gamma$ is the tensor field \eqref{gamma}, which by \eqref{lcc2} has the form 
$\Gamma_\alpha\gamma =   \ell^*(\alpha)\gamma' - \alpha\, \ell(\gamma')$. Hence,
 \eqq{
\Gamma_\alpha(\beta\circ\gamma) &= \ell^*(\alpha)(\beta\circ\gamma)' - \alpha\ell\bra{(\beta\circ\gamma)'}\\
&= \ell^*(\alpha)\ell(\la'\beta)'\gamma +  \ell^*(\alpha)\ell(\la'\beta)\gamma' + \ell^*(\alpha)\beta'\ell(\la'\gamma) 
+ \ell^*(\alpha)\beta\ell(\la'\gamma)'\\
&\quad -\alpha\ell\bra{\ell(\la'\beta)'\gamma} -\alpha\ell\bra{\ell(\la'\beta)\gamma'}
-\alpha\ell\bra{\beta'\ell(\la'\gamma)}-\alpha\ell\bra{\beta\ell(\la'\gamma)'} 
}
and
\eqq{
\Gamma_\alpha\beta\circ\gamma &= \ell\bra{\la'\Gamma_\alpha\beta}\gamma 
+ \Gamma_\alpha\beta\ell(\la'\gamma)\\
&= \ell\bra{\ell^*(\alpha)\la'\beta'}\gamma - \ell\bra{\la'\alpha\ell(\beta')}\gamma + \ell^*(\alpha)\beta'\ell(\la'\gamma)
- \alpha\ell(\beta')\ell(\la'\gamma).
}
Recall that $\Dir_{\alpha^\sharp}\la' = (\alpha^\sharp)'$, where
$\alpha^\sharp\in\vf{\Alg_M}$ is given by \eqref{iso}. Thus,
\eqq{
(\Dir_{\alpha^\sharp}\circ)(\beta,\gamma) &= \ell\bra{\Dir_{\alpha^\sharp}\la'\beta}\gamma 
+ \beta\ell\bra{\Dir_{\alpha^\sharp}\la'\gamma}
= \ell\bra{\alpha^\sharp{}'\beta}\gamma + \beta\ell\bra{\alpha^\sharp{}'\gamma}\\
&= \ell\bra{\ell^*(\alpha)\la''\beta}\gamma - \ell\bra{\ell(\alpha')\la'\beta}\gamma + \ell\bra{\ell(\la'\alpha)'\beta}\gamma\\
&\quad + \beta\ell\bra{\ell^*(\alpha)\la''\gamma} - \beta\ell\bra{\ell(\alpha')\la'\gamma} + \beta\ell\bra{\ell(\la'\alpha)'\gamma}.
}
Now, substituting the above formulae to \eqref{cdm} we have
\eqq{
&\bra{\nabla_{\alpha^\sharp}\circ}(\beta,\gamma) - \bra{\nabla_{\beta^\sharp}\circ}(\alpha,\gamma) =
\bra{\nabla_{\alpha^\sharp}\circ}(\beta,\gamma) - \pobr{\alpha\leftrightarrow \beta}\\
&\qquad= \ell\bra{\ell^*(\alpha)\la''\beta}\gamma - \cancel{\ell\bra{\ell(\alpha')\la'\beta}\gamma} + \ell\bra{\ell(\la'\alpha)'\beta}\gamma +\beta\ell\bra{\ell^*(\alpha)\la''\gamma}\\
&\qquad\quad - \beta\ell\bra{\ell(\alpha')\la'\gamma} + \cancel{\beta\ell\bra{\ell(\la'\alpha)'\gamma}}- \ell^*(\alpha)\ell(\la'\beta)'\gamma -  \cancel{\ell^*(\alpha)\ell(\la'\beta)\gamma'}\\
&\qquad\quad - \cancel{\ell^*(\alpha)\beta'\ell(\la'\gamma)}- \ell^*(\alpha)\beta\ell(\la'\gamma)' +\cancel{\alpha\ell\bra{\ell(\la'\beta)'\gamma}} +\alpha\ell\bra{\ell(\la'\beta)\gamma'}\\
&\qquad\quad+\alpha\ell\bra{\beta'\ell(\la'\gamma)}+\alpha\ell\bra{\beta\ell(\la'\gamma)'} 
+ \ell\bra{\ell^*(\alpha)\la'\beta'}\gamma - \cancel{\ell\bra{\la'\alpha\ell(\beta')}\gamma}\\
&\qquad\quad + \cancel{\ell^*(\alpha)\beta'\ell(\la'\gamma)} - \alpha\ell(\beta')\ell(\la'\gamma)
+ \beta\ell\bra{\ell^*(\alpha)\la'\gamma'} - \beta\ell\bra{\la'\alpha\ell(\gamma')}\\
&\qquad\quad + \cancel{\ell^*(\alpha)\ell(\la'\beta)\gamma'} - \alpha\ell(\la'\beta)\ell(\gamma') - \pobr{\alpha\leftrightarrow \beta}, %  \\
}
where $\pobr{\alpha\leftrightarrow \beta}$ stands for all the remaining terms arising from the permutation
of $\alpha$ with $\beta$ in the preceding terms. Some terms cancel out as in the preceding terms we are allowed
to  permute $\alpha$ and $\beta$ with simultaneous change of sign. Using that property we can assort all terms obtaining 
\eqq{
&\bra{\nabla_{\alpha^\sharp}\circ}(\beta,\gamma) - \bra{\nabla_{\beta^\sharp}\circ}(\alpha,\gamma) =\\
&\qquad = \brac{\ell\bra{\ell^*(\alpha)(\la'\beta)'} - \ell\bra{\alpha\ell(\la'\beta)'} - \ell^*(\alpha)\ell(\la'\beta)'}\gamma\\
&\qquad\quad +\beta \brac{\ell\bra{\ell^*(\alpha)(\la'\gamma)'} - \ell\bra{\alpha\ell(\la'\gamma)'} - \ell^*(\alpha)\ell(\la'\gamma)'}   \\
&\qquad\quad -\beta \brac{\ell\bra{\ell(\alpha')\la'\gamma} + \ell\bra{\alpha'\ell(\la'\gamma)} - \ell(\alpha')\ell(\la'\gamma)}\\
&\qquad\quad +\alpha \brac{\ell\bra{\ell(\la'\beta)\gamma'} + \ell\bra{\la'\beta\ell(\gamma')} - \ell(\la'\beta)\ell(\gamma')}
- \pobr{\alpha\leftrightarrow \beta}\\
&\qquad= -\kappa\, \alpha (\la'\beta)'\gamma -\kappa\, \alpha\beta (\la'\gamma)' -\kappa\, \la'\alpha'\beta\gamma
+ \kappa\, \la'\alpha\beta\gamma' - \pobr{\alpha\leftrightarrow \beta} = 0,
}
where the identities \eqref{rel3} and \eqref{RBi} were used. This completes the proof.
\end{proof}

\begin{remark}
The above proof significantly simplifies if the endomorphism $\ell$
is antisymmetric ($\ell^* = -\ell$) and it commutes with the derivation $\pr$, that is $\ell'=0$.
\end{remark}

\subsection{The recurrence formula}

In particular, it would be highly advantageous to use Proposition~\ref{Fpro} for 
a derivation of pre-potentials on Frobenius manifolds which could be obtained from the above scheme. 
This requires use of the recurrence formula \eqref{rec2}. Let see how it fits in the above scheme
under some additional assumption.  

\begin{proposition}
Assume that the differential $1$-forms $d\H^k_{(n)}:\Alg_M\arrow\Alg$ are differentiable functions of a single variable $\la\in\Alg_M$\footnote{We can make such assumption since with each differential $1$-form
on $\Alg_M$  we can identify a differentiable map from $\Alg_M$ to $\Alg$, see Remark~\ref{rde}.} such that
\eqq{
\bra{d\H^k_{(n)}}' = \frac{d (d\H^k_{(n)})}{d \la}\la' \qquad\text{and}\qquad \Dir_X d\H^k_{(n)} 
= \frac{d (d\H^k_{(n)})}{d \la}X, 
}
where $X\in\vf{\Alg_M}$. Then, the recurrence formula \eqref{rec2} takes the particular simple form
\eq{\label{recH}
	\frac{d (d\H^k_{(n)})}{d \la} = d\H^k_{(n-1)}.
}	
\end{proposition}
\begin{proof}
From the assumptions and formulae \eqref{fmult} and \eqref{lcc2} we have
\eqq{
	\nabla_{\alpha^\sharp}d\H^k_{(n)} &= \frac{d (d\H^k_{(n)})}{d \la}\alpha^\sharp
	+ \alpha\, \ell\bra{(d\H^k_{(n)})'} - \ell^*(\alpha)(d\H^k_{(n)})' \\
	&= \ell(\la'\alpha)\frac{d (d\H^k_{(n)})}{d \la} + \alpha\, \ell\Big(\la'\frac{d (d\H^k_{(n)})}{d \la}\Big)
}
and
\eqq{
\alpha\circ d\H^k_{(n-1)} = \ell(\la' \alpha)d\H^k_{(n-1)}  + \alpha\,\ell(\la'd\H^k_{(n-1)}),
}
which substituted to \eqref{rec2} give \eqref{recH}.
\end{proof}

\begin{proposition}
For arbitrary $\alpha,\beta\in\of{\Alg_M}$ the following relation is valid:
\eq{\label{arel}
\bra{\alpha\circ\beta}^\sharp = \alpha^\sharp\ell(\la'\beta) + \la'\ell^*(\alpha^\sharp\beta),
}
where $\alpha^\sharp\in\vf{\Alg_M}$ is given by \eqref{iso}.
\end{proposition}
The proof is straightforward using the identities \eqref{RBi} and \eqref{dRB}.

\section{Frobenius manifolds on the space of meromorphic functions}\label{mer}

\subsection{Algebra of meromorphic functions.}\label{ss1}

Our aim is to illustrate the scheme for construction of Frobenius manifolds by applying it to the algebra 
of meromorphic functions on the Riemann sphere $\CP$. So, let
\eqq{
\Alg = \pobr{f:\hat{\Cm}\map\Cm \,|\, \text{$f$ is meromorphic}}, 
}
where $\hat{\Cm}\cong \CP$ is the extended complex plane. Let $p$  denote the variable in $\hat{\Cm}$.
The space $\Alg$ is an infinite dimensional algebra with respect to the commutative associative multiplication
of functions. 

We define the derivation $\pr\in\Der{\Alg}$ for $s=0$ or $s=1$ by the formula
\eq{\label{der}
f' \equiv \pr f := p^s\pd{f}{p}\qquad f\in\Alg.
}
We will consider this two different cases of $s=0,1$ simultaneously.

Our aim is to define Frobenius algebras taking advantage of expansions near some marked points $\nu\in\hat{\Cm}$ on the extended complex plane. We will distinguish three types of them. The first one is a fixed point at infinity, i.e. $\nu=\infty$.
The second one is a fixed finite point, without loss of generality we take $\nu=0$. The last one is a 'finite' not fixed point, which can vary over the complex plane, i.e.  $\nu=v\in\Cm$. Important is fact that $v$ can be taken as one of the coordinates on some of the underlying manifolds that will be considered.

At a given marked point $\nu$ the trace form $\Tr_\nu:\Alg\arrow \Cm$ is defined by
\eq{\label{traces}
\Tr_\nu (f) := \epsilon \res_{p=\nu} \bra{p^{-s} f}\qquad s=0,1,
}
where $\epsilon =-1$ for $\nu=\infty$ and $\epsilon =1$ otherwise, and we take advantage of the standard residue. The invariance of \eqref{traces} with respect to the derivation \eqref{der} is a straightforward consequence of the integration by parts.

\subsection{The endomorphisms} 

Meromorphic functions can be expanded into Laurent series at the marked points and
at these points we can define projections on the finite parts of the series.  Therefore, for any $f\in\Alg$ 
at $\nu=\infty$ let
\eqq{
\brac{f(p)}^\infty_{\me l}:= \sum_{i\me l}a_ip^i\quad\text{for}\quad f(p)= \sum_{i\les n} a_i p^i,
}
where $n=\deg_\infty f$ and $\res_{p=\infty} f = -a_{-1}$. For finite $\nu=0,v$: 
\eqq{
\brac{f(p)}^\nu_{< l}:= \sum_{i< l}a_i(p-\nu)^i\quad\text{for}\quad f(p) = \sum_{i\me-m} a_i (p-\nu)^i,
}
where $\deg_\nu f = m$ and $\res_{p=\nu}f  = a_{-1}$. Respectively, the following computations
will be carried by means of the spaces of formal Laurent series. 

We will use the following notation:\footnote{See the footnote~\ref{foot}.}
\eqq{
\Alg^\infty = \Cm\cc{p^{-1}}\quad\text{for}\quad \nu=\infty
}
and
\eqq{
\Alg^\nu = \Cm\cc{p-\nu}\quad\text{for}\quad \nu=0,v;
} 
for the corresponding spaces of formal Laurent series at $\infty$ and finite $\nu$.
The respective `regular' dual algebras are:
\eqq{
	\tAlg^{\infty}\equiv \Cm\cc{p}\cong(\Alg^{\infty})^* 
}
and
\eqq{
\tAlg^{\nu}\equiv \Cm\cc{(p-\nu)^{-1}} \cong(\Alg^{\nu})^*;
} 
which are defined with respect to the duality pairing given by means of the trace form~\eqref{traces}, that is $\Alg^\nu\times\tAlg^\nu\arrow\Cm$, s.t.~$(X,\gamma)\map \tr_\nu(X\gamma)$. Let's define the following subspaces of $\Alg^\nu$:
\eqq{
	\Alg^\infty_{\les k} := \pobr{f\in\Alg^\infty\ |\ \deg_\infty f\les k}
}
for $\nu=\infty$ and
\eqq{
	\Alg^\nu_{> k} = \pobr{f\in\Alg\ |\ \deg_\nu f< -k}
}
for $\nu = 0$ or $v$. The 'dual' respective subspaces of $\tAlg^\nu$ are defined in a similar way.

We define $\ell\in\End{\Alg}$ by the formula
\eq{\label{ell}
	\ell(f) = p^{s}\brac{p^{-s}f}_{\me 0}^\nu  - \frac{1}{2}f = \frac{1}{2}f - p^{s}\brac{p^{-s}f}_{<0}^\nu \qquad s=0,1.
}

\begin{proposition}\label{propll}
The endomorphisms \eqref{ell} satisfy the Rota-Baxter identity \eqref{RBi} for $\kappa=\frac{1}{4}$. 
The adjoints of \eqref{ell} with respect to the traces \eqref{traces} have the form
\eq{\label{elld}
		\ell^*(f) = \frac{1}{2}f - \brac{f}^\nu_{\me 0} =  \brac{f}^\nu_{< 0} -\frac{1}{2}f\qquad f\in\Alg.
}
Moreover, the relation \eqref{frel} is satisfied by endomorphisms \eqref{ell}.
\end{proposition}
\begin{proof}
For $s=0$ and $s=1$ the spaces $\Alg^\nu$ can be decomposed into direct sums of subalgebras, that is $\Alg^\nu= \Alg^\nu_{\me s} \oplus \Alg^\nu_{< s}$. Hence, by Proposition~\ref{rbs} the linear maps 
\eqq{
\tilde{\ell}(f) = \brac{f}_{\me s}^\nu - \frac{1}{2} f= \frac{1}{2}f - \brac{f}_{<s}^\nu\qquad f\in\Alg^\nu 
} 
satisfy the identity \eqref{RB} and \eqref{RBi}.  These maps  coincide with \eqref{ell} for all cases but one. 

The exception occurs for finite $\nu=v$ and $s=1$. It must be considered separately. In this case 
\eqref{ell} can be written in the form
\eq{\label{ell2}
		\ell(f) = \brac{f}_{\me 1}^v - \frac{1}{2}f + v \brac{p^{-1}f}^v_0\equiv \tilde{\ell}(f) + P(f),
}
where $\brac{\sum_ia_i(p-v)^i}^v_0:= a_0$ and $P(f) = v[p^{-1}f]^v_0$. 

Assume that $\tilde{\ell}$ satisfies \eqref{RBi}. Let $P$ be such that $\tilde{\ell}(P(f)g) = P(f)\tilde{\ell}(f)$
and $P(P(f)g) = P(f)P(g)$. These requirements are satisfied by \eqref{ell2}. 
Then, $\ell = \tilde{\ell} + P$ satisfies \eqref{RBi} for the same $\kappa$ iff
\eq{\label{pl}
P\bra{\tilde{\ell}(f)g}+ P\bra{f\tilde{\ell}(g)} + P(f)P(g) = 0. 
}
The identity \eqref{pl}, for the map \eqref{ell2}, can be showed using the following relations:
\eqq{
		\brac{p^{-1}\brac{f}^v_{\me 1}g}^v_0 &= \brac{f\brac{p^{-1}g}^v_{<0}}^v_0\\
	\brac{p^{-1}f\brac{g}^v_{\me 1}}^v_0 &= \brac{f\brac{p^{-1}g}^v_{\me 0}}^v_0 
	- v\brac{p^{-1}f}^v_0\brac{p^{-1}g}^v_0. 	 
}

To find $\ell^*$ it is sufficient to observe that
\eqq{
	\tr_\nu\bra{p^s\brac{p^{-s}f}^\nu_{\me 0}g} = \res_{p=\nu}\bra{\brac{p^{-s}f}^\nu_{\me 0}g}
	= \res_{p=\nu}\bra{p^{-s}f\brac{g}^\nu_{<0}} =  \tr_\nu\bra{f\brac{g}^\nu_{<0}}.
}

The last statement of the proposition is straightforward as
\eqq{
\ell(f') +\ell^*(f)' = p^s\brac{\pd{f}{p}}^\nu_{\me 0} - p^s\pd{\brac{f}^\nu_{\me 0}}{p}  = 0.
}
This finishes the proof.
\end{proof}

\begin{remark}
The construction of Frobenius manifolds in this section corresponds to the construction of bi-Hamiltonian 
structures for dispersionless systems with meromorphic Lax representations presented in \cite{Szab},
where the classical $r$-matrix formalism is applied to commutative algebras equipped with Poisson bracket
$\pobr{\cdot,\cdot}_s =  p^s\pr_p\wedge\pr_x$.  The related classical $r$-matrices
have the form $r = P^\nu_{\me k-s} - \frac{1}{2}$ and are classified with respect to $s,k\in\Z$. The condition \eqref{frel} is crucial for Theorem~\ref{main}  to hold. The only relevant cases of $r$, up to equivalence, satisfying \eqref{frel} with $\ell=r^*$ are for $s=k=0$ and $s=k=1$, and this cases are related to dispersionless systems of the KdV and Toda type, respectively.
\end{remark}

\subsection{Frobenius structure on the spaces of formal Laurent series}\label{ss3}

First we will show how to construct Frobenius algebras and formal pre-Frobenius manifolds 
on appropriate submanifolds $\Alg^\nu_M$ of the spaces of formal Laurent series. Only then it will be natural 
  to extend (reduce) this formalism to underlying manifolds made of meromorphic functions.

We will consider the situation when $\Alg^\nu_M$ consist of formal Laurent series $\la(p)$ at $\nu$ with prescribed 
order and form. For $\nu=\infty$ we define 
\eq{\label{ami}
\Alg_M^\infty = \pobr{\la(p) = p^n + u_{n-1}p^{n-1} + u_{n-2}p^{n-2} + \ldots\, |\, u_i\in\Cm, \text{$u_{n-1}=0$ for $s=0$}},
}
where the degree $\deg_\infty\la=n$ is fixed and for finite $\nu$ we define
\eqq{
\Alg_M^\nu = \pobr{\la(p)=\ldots + u_{1-m}(p-\nu)^{1-m} +  u_{-m}(p-\nu)^{-m}\, |\, u_i\in\Cm\text{ and } v\in\Cm \text{ (if $\nu=v$)}},
}
where $\deg_\nu=m$ is also fixed. The complex coefficients $u_i$ and $v$ (when $\nu=v$) are coordinates on the underlying 'infinite-dimensional' manifolds associated to $\Alg^\nu_M$. The tangent spaces are spanned
by derivations of $\la(p)$ with respect to these coordinates, that is 
\eq{\label{tgsp}
\T_\la\Alg_M^\nu= \spn\pobr{\frac{\pr\la}{\pr u_i},\frac{\pr\la}{\pr v}\,\text{(if $\nu=v$)}}
} 
at point $\la\in\Alg_M^\nu$. The cotangent spaces $\T_\la^*\Alg_M^\nu$ are dual spaces with respect to the corresponding  trace forms \eqref{traces} defined so that $\T_\la^*\Alg^\nu_M\subset \tAlg^\nu$.

One finds that for $\nu=\infty$:
\eqq{
	\T_\la\Alg^\infty_M \cong\Alg^\infty_{\les n+s-2}\quad\Rightarrow\quad 
	\T_\la^*\Alg^\infty_M \cong    \tAlg^\infty_{\me 1-n},
}
for $\nu=0$:
\eqq{
	\T_\la\Alg^0_M \cong\Alg^0_{\me -m}\quad\Rightarrow\quad 
	\T_\la^*\Alg^0_M \cong  \tAlg^0_{< m+s},
}  
and for $\nu=v$:
\eqq{
		\T_\la\Alg^v_M\cong \begin{cases}
      \Alg^v_{\me -m-1}& \text{for $m\neq 0$} \\
      \Alg^v_{\me 0} & \text{for $m=0$}
\end{cases}
\quad\Rightarrow\quad 
		\T_\la^*\Alg^v_M\cong \begin{cases}
    \tAlg^v_{\les m}& \text{for $m\neq 0$} \\
     \tAlg^v_{< 0} & \text{for $m=0$}.
\end{cases}
}

\begin{proposition}
The derivation \eqref{der} and endomorphisms \eqref{ell} are invariant on the underlying manifolds $\Alg^\nu_M$.
\end{proposition}

To show the above proposition it is enough to check that the defined above $\pr$ and $\ell$ commute with the derivations with respect to coordinates on $M$ spanning the tangent spaces \eqref{tgsp}.

Using \eqref{ell}, the multiplication \eqref{fmult} in $\Alg$ takes the form
\eq{\label{mn}
\begin{split}
	\alpha\circ \beta &= p^s\brac{\la_p \alpha}^\nu_{\me 0}\beta + p^s\alpha\brac{\la_p \beta}^\nu_{\me 0} 
	- p^s\la_p \alpha\beta\\
	&=  p^s\la_p \alpha\beta - p^s\brac{\la_p\alpha}^\nu_{< 0}\beta - p^s\alpha\brac{\la_p\beta}^\nu_{< 0}.
\end{split}
}
Our intention is to use \eqref{mn} to define multiplication in the respective cotangent bundles.
When $\alpha, \beta, \lambda\in\Alg$ are meromorphic functions the above product is well-defined 
and  the multiplication is closed in $\Alg$.  This fact will be important in Section~\ref{frobsec}. 
However, when we deal with formal Laurent series the multiplication \eqref{mn}, in general, is ill-defined,  since the product $\la_p \alpha\beta$ yields doubly-infinite series, so the multiplication in this case is not closed, in particular,
in the (dual) Laurent spaces $\tAlg^\nu$.  

One can resolve the above issue imposing some additional restrictions. For instance, we could restrict ourselves defining the underlying manifold subspaces $\Alg_M^\nu$ and their tangent and cotangent spaces to deal only with `nonformal' Laurent series, which are always expansions of some meromorphic functions. In this case \eqref{mn} would be closed in $\Alg$. Alternatively we could assume that the elements $\lambda\in\Alg_M^\nu$ have only finitely many nonzero coefficients, which in this case is sufficient for the multiplication to be closed in~$\tAlg^\nu$. 

Nevertheless, to stay in the more general setting we can project the results of the operation \eqref{mn} on the respective cotangent spaces $\T_\la^*\Alg_M^\nu$.  After such projection the multiplication preserves associativity if
\eq{\label{ideal2}
		f\circ (\T_\la\Alg^\nu_M)^\perp \subset (\T_\la\Alg^\nu_M)^\perp
}
holds for arbitrary $f\in \tAlg^\nu$ or $\Alg^\nu$. If the above relation holds it means that $(\T_\la\Alg^\nu_M)^\perp$ is `formally' an ideal with respect to the multiplication \eqref{mn} within the space of doubly-infinite Laurent series. Since in this case the multiplication is ill-defined we cannot  use directly the related quotient structure to define associative algebra in $\T_\la^*\Alg^\nu_M$.  However, the projection eliminates this issue and is legitimate when  \eqref{ideal2} is valid.

\begin{proposition}\label{multproj}
\begin{itemize}
\item[]
\item For $\nu=\infty$ and $s=0,1$ the projection of the multiplication \eqref{mn} on the cotangent bundle $\T_\la^*\Alg_M^\infty$ 
is legitimate for arbitrary $\deg_\infty\la$ and after the projection it takes the form:
\eqq{
	\alpha\circ \beta &= p^s\brac{\alpha\brac{\la_p \beta}^\infty_{\me 0} + \brac{\la_p \alpha}^\infty_{\me 0}\beta
	- \la_p \alpha\beta}^\infty_{\me 1-n-s}.
}
\item For $\nu=0$ the projection of \eqref{mn} on $\T_\la^*\Alg_M^0$ is legitimate in the case of $s=0$
only if $\deg_0\la=0$ and in the case of $s=1$ for arbitrary $\deg_0\la$. In these  cases we obtain
\eqq{
	\alpha\circ \beta =  p^s\brac{\la_p \alpha\beta - \brac{\la_p\alpha}^0_{< 0}\beta - \alpha\brac{\la_p\beta}^0_{< 0}}^0_{<m} .
}
\item For $\nu=0$ and $s=0,1$ the projection of \eqref{mn} on $\T_\la^*\Alg_M^v$ is legitimate for arbitrary $\deg_v\la$
and after projection the multiplication takes the form
\eqq{
\alpha\circ \beta =  
\begin{cases}
\brac{p^s\la_p \alpha\beta - p^s\brac{\la_p\alpha}^v_{< 0}\beta - p^s\alpha\brac{\la_p\beta}^v_{< 0}}^v_{\les m} &\text{for $m\neq 0$}\\
\brac{p^s\la_p \alpha\beta - p^s\brac{\la_p\alpha}^v_{< 0}\beta - p^s\alpha\brac{\la_p\beta}^v_{< 0}}^v_{< 0} &\text{for $m= 0$}.
\end{cases}
}
\end{itemize}
In all the above cases the related multiplication is abelian and associative in the respective cotangent spaces $\T_\la^*\Alg^\nu_M$.
\end{proposition}
\begin{proof}
Let's first consider the case of $\nu=\infty$. The orthogonal complement to the tangent space $\T_\la\Alg^\infty_M$ is given by $(\T_\la\Alg^\infty_M)^\perp\cong\Alg^\infty_{\les-n}$.  Since $\la_p\in\Alg^\infty_{<n}$ it follows that  $\brac{\la_p\Alg^\infty_{\les-n}}^\infty_{< 0} = \la_p\Alg^\infty_{\les-n}$.
Hence, for arbitrary $f$ from $\Alg^\infty$ or $\tAlg^\infty$ and $s=0,1$:
\eqq{
	f\circ \Alg^\infty_{\les-n} = - p^s\brac{\la_pf}^\infty_{< 0}\Alg^\infty_{\les-n}\subset \Alg^\infty_{<s-n}\subset \Alg^\infty_{\les-n} . 
}
This means that $\Alg^\infty_{\les-n}$ is 'formally' an ideal in the space of doubly-infinite Laurent series
and that the projection on the cotangent space $\T_\la^*\Alg_M^\infty\cong   \tAlg^\infty_{\me 1-n}$ is legitimate. 
 
 For $\nu=0$ we have $(\T_\la\Alg^0_M)^\perp \cong \Alg^0_{\me m+s}$, $\la_p \in\Alg^0_{\me -m-1}$ for $m\neq 0$ and $\la_p \in\Alg^0_{\me 0}$ for $m=0$. Thus,
\eqq{
	\brac{\la_p\Alg^0_{\me m+s}}^0_{\me 0} = \begin{cases}
\la_p\Alg^0_{\me m+s} +(s-1)\brac{\la_p\Alg^0_{\me m+s}}^0_{-1}p^{-1} & \text{for $m\neq 0$}\\
\la_p\Alg^0_{\me s}  & \text{for $m=0$}
\end{cases}
}
and consequently for arbitrary $f$ from $\Alg^0$ or $\tAlg^0$:
\eqq{
	f\circ \Alg^0_{\me m+s} = \begin{cases}
p^s\brac{\la_pf}^0_{\me 0}\Alg^0_{\me m+s} +(s-1)p^{s-1}f\brac{\la_p\Alg^0_{\me m+s}}^0_{-1} & \text{for $m\neq 0$}\\
p^s\brac{\la_pf}^0_{\me 0}\Alg^0_{\me s}  & \text{for $m=0$}.
\end{cases}
}
Hence, we have $f\circ \Alg^0_{\me m+s}\subset\Alg^0_{\me m+s}$ for $s=0$ only if $m=0$ and for $s=1$ and arbitrary $m$. Thus, only in these cases we are allowed to project the multiplication \eqref{mn} on $\T_\la^*\Alg^0_M \cong  \tAlg^0_{< m+s}$.

For $\nu=v$ the complement $(\T_\la\Alg^v_M)^\perp \cong \Alg^v_{> m}$ for $m\neq 0$ and
$(\T_\la\Alg^v_M)^\perp \cong \Alg^v_{\me 0}$ for $m= 0$. Since $\la_p \in\Alg^v_{\me -m-1}$ for $m\neq 0$ and $\la_p \in\Alg^v_{\me 0}$ for $m=0$, we have
\eqq{
\brac{\la_p(\T_\la^*\Alg^v_M)^\perp}^v_{\me 0} = \la_p(\T_\la^*\Alg^v_M)^\perp.
} 
Hence, for arbitrary $f$ from $\Alg^v$ or $\tAlg^v$:
\eqq{
f\circ (\T_\la\Alg^v_M)^\perp = p^s\brac{\la_pf}^v_{\me 0}(\T_\la\Alg^v_M)^\perp \subset (\T_\la\Alg^v_M)^\perp,
}
which holds without any further restrictions.
\end{proof}

If we restricts ourselves to deal only with `nonformal' Laurent series it follows that for all the cases from Proposition~\ref{multproj}  the orthogonal complements  $(\T_\la\Alg^\nu_M)^\perp$ are ideals in $\Alg$ with respect to the multiplication \eqref{mn} and that the
multiplication in the cotangent bundles $\T_\la^*\Alg^\nu_M$ can be defined by means of the quotient structure $\T_\la^*\Alg^\nu_M\cong \Alg/(\T_\la\Alg^\nu_M)^\perp$. In the alternative restricted case, when only finitely many coefficients in $\la\in\Alg_M^\nu$ are nonzero,  it follow that for the cases of the above proposition $\tAlg^\nu\cap (\T_\la\Alg^\nu_M)^\perp$ are ideals in $\tAlg^\nu$ with respect \eqref{mn} and that we can define the multiplication in $\T_\la^*\Alg^\nu_M$ by means of the quotient structure such that $\T_\la^*\Alg^\nu_M\cong \tAlg^\nu/(\tAlg^\nu\cap (\T_\la\Alg^\nu_M)^\perp)$.

The respective contravariant metric is given on each $\Alg_M^\nu$ by the formula \eqref{fmetric}, that is 
\eq{\label{met}
      \eta^*(\alpha,\beta) = \tr_{\nu}\bra{\alpha\circ\beta} = \tr_{\nu}\bra{\alpha^\sharp\beta}\qquad \alpha,\beta\in\of{\Alg_M^\nu},
}
where \eqref{mn} is used.\footnote{Here, necessarily, the multiplication in the form \eqref{mn}, that is not projected on the 
tangent spaces, must be used. Otherwise, in some cases, the outcome of the trace form might identically vanish, which would be incorrect.}	 
The related canonical isomorphism $\sharp:\T_\la^*\Alg_M^\nu\arrow \T_\la\Alg_M^\nu$ at $\la\in\Alg^\nu_M$ 
is such that
\eq{\label{sharp}
	\alpha^\sharp = p^s\brac{\la_p \alpha}^\nu_{\me 0} - p^s\la_p \brac{\alpha}^\nu_{\me 0} 
	= p^s\la_p \brac{\alpha}^\nu_{< 0} - p^s\brac{\la_p \alpha}^\nu_{< 0}.
}

The cotangent spaces $\T_\la^*\Alg^\nu_M$ are defined modulo orthogonal complements
to the tangent spaces $\T_\la\Alg^\nu_M$. Thus we must verify when the metric is consistent
with the related quotient structures, that is we must check when the following relation holds:
\eqq{
		\tr_\nu \bra{f\circ (\T_\la\Alg^\nu_M)^\perp} = 0,
}
where $f\in\tAlg^\nu$ or $\Alg^\nu$.

\begin{proposition}\label{qstr}
The metric \eqref{met} is consistent with the quotient structure of cotangent spaces $\T_\la^*\Alg^\nu_M$:
\begin{itemize}
\item for $\nu=\infty$ and $s=0,1$ if $\deg_\infty\la \me 1$;
\item for $\nu=0$ and $s=0$ if $\deg_0\la =0$;
\item for $\nu=0$ and $s=1$ if $\deg_0\la\me -1$;
\item for finite $\nu=v$ and $s=0,1$ if $\deg_v\la \me -1$.
\end{itemize}
\end{proposition} 
\begin{proof}
For $\nu=\infty$ we have 
\eqq{
	\tr_\infty\bra{f\circ (\T_\la\Alg^\infty_M)^\perp} = -\res_{p=\infty}\Alg^\infty_{< -n}= 0 \quad\iff\quad n\me 1,
}
where $f\in \tAlg^\nu$ or $\Alg^\nu$. 

Let $\nu = 0$.  For $s=0$ and $m\neq 0$ one observes that $\tr_0\bra{f\circ (\T_\la\Alg^0_M)^\perp}$
is in general different from zero. in the remaining cases we see that
\eqq{
	\tr_0\bra{f\circ (\T_\la\Alg^0_M)^\perp} = 
\Alg^0_{\me m+s} = 0 \quad\iff\quad m\me -s,
}
where $f\in \tAlg^0$ or $\Alg^0$.

Analogously, for $\nu=v$ one finds that 
\eqq{
	\tr_0\bra{f\circ (\T_\la\Alg^v_M)^\perp} = 0 \quad\iff\quad m\me -1,
}
which finishes the proof. 
\end{proof}

\begin{proposition}\label{ndm} 
The metric \eqref{met}, at a generic point $\la\in\Alg_M^\nu$, is nondegenerate:
\begin{itemize}
\item for $\nu=\infty$ and $s=0,1$ if $\deg_\infty\la \me 1$;
\item for $\nu=0$ and $s=0$ if $\deg_0\la =0$;
\item for $\nu=0$ and $s=1$ if $\deg_0\la\me 1$ or $\deg_0\la=-1$;
\item for finite $\nu=v$ and $s=0,1$ if $\deg_v\la \me -1$.
\end{itemize}
In the remaining cases the metric is always degenerate or not well defined on $\Alg_M^\nu$.
Moreover, in the all above cases the metric \eqref{met} is compatible with
the multiplication \eqref{mn}, that is 
\eqq{
	\eta(\alpha,\beta\circ\gamma) = \eta(\alpha\circ\beta,\gamma)
}
for arbitrary $\alpha,\beta,\gamma\in\of{\Alg_M^\nu}$.

\end{proposition}
\begin{proof} Let $\la\in\Alg^\infty_M$, then $\la_p\in\Alg_{\les n-1}^\infty$. Thus for $\nu=\infty$
the map \eqref{sharp} has the image:   
\eqq{
\im\sharp =\Alg^\infty_{\les n+s-2}\oplus \Alg^\infty_{\les s-1}.
}
Similarly, for $\la\in\Alg_M^0$ we have $\la_p\in\Alg_{\me -m-1}^0$ for $m\neq 0$ and
$\la_p\in\Alg_{\me 1}^0$ for $m=0$. Hence, for $\nu=0$:
\eqq{
\im\sharp = \begin{cases}
 \Alg^0_{\me s}\oplus \Alg^0_{\me -m+s-1}     & \text{for $m\neq 0$} \\
 \Alg^0_{\me s}     & \text{for $m=0$}.
\end{cases}
}
For $\la\in\Alg_M^v$ we have $\la_p\in\Alg_{\me -m-1}^v$ for $m\neq 0$ and
$\la_p\in\Alg_{\me 0}^v$ for $m=0$. Since $p^s\in\Alg^v_{\me 0}$, for finite $\nu= v$ we have
\eqq{
\im\sharp= \begin{cases}
 \Alg^v_{\me 0}\oplus \Alg^v_{\me -m-1}   & \text{for $m\neq 0$} \\
\Alg^v_{\me 0}     & \text{for $m=0$}.
\end{cases}
}
Requiring that $\im\sharp = \T_\la\Alg^\nu_M$, we get the conditions on the degrees
of $\la\in\Alg^\nu_M$. Besides, in all these cases $\ker\sharp = \emptyset$ at a generic point $\la\in\Alg_M^\nu$.

In all the above cases the multiplication \eqref{mn} is well-defined in the cotangent bundles $\T^*\Alg^\nu_M$
by means of Proposition~\ref{multproj} as well as the metric \eqref{met} is also consistent with the quotient structures of $\T_\la^*\Alg^\nu_M$, Proposition~\ref{qstr}. As result in these cases the metric \eqref{met} is compatible with the multiplication \eqref{mn}. 
\end{proof}

\begin{lemma}\label{uvf}
In each case from the Proposition~\ref{ndm}
there is a Frobenius algebra structure defined in the cotangent bundles  $\T^*\Alg^\nu_M$.
The respective unit vector fields are given by
\eq{\label{unitvf}
e = 
\begin{cases}
1-\la_p & \text{for $\nu=\infty$, $s=0$ and $\deg_\infty\la = 1$,}\\
1-\frac{1}{u_1}\la_p &  \text{for $\nu=0$, $s=1$ and $\deg_0\la = -1$,}\\
1 & otherwise.
\end{cases}
} 
So, almost always $e=1$. In all the above cases but one the respective unit vector fields are flat. The exception is the case of $\nu=0$, $s=1$ and $\deg_0\la= -1$.
\end{lemma}
\begin{proof}
In fact, the first part of the lemma is a corollary to Propositions~\ref{multproj}-\ref{ndm}. 
It is only left to show that in each case there exists unit element ($1$-form) for the respective multiplication
\eqref{mn}. We will take advantage from the nondegeneracy of the metric \eqref{met} and we will use the canonical isomorphism \eqref{sharp}.
Notice that all the unit vector fields \eqref{unitvf} belong to the respective tangent spaces. 
Let  $\eps$ be $1$-form such that $\eps^\sharp =e$. By relation \eqref{arel} for arbitrary $\beta\in\of{\Alg_M^\nu}$
we have 
\eqq{
	(\eps\circ\beta)^\sharp = p^s e\brac{\la_p\beta}^\nu_{\me 0} -p^s\la_p\brac{e\beta}^\nu_{\me 0} = p^s\la_p\brac{e\beta}^\nu_{< 0} - p^s e\brac{\la_p\beta}^\nu_{< 0}. 
}
In the case when $e=1$ one finds immediately  that  $(\eps\circ\beta)^\sharp =\beta^\sharp$.
In the remaining cases it is slightly more involved to check. This shows that $1$-forms $\eps$ are units in
the respective quotient algebras.

The co-unity $\eps$ is closed on $\Alg^\nu_M$, which is consequence of Proposition~\ref{propll} and
Proposition~\ref{prop}. Therefore, for the flatness of $e$ it is sufficient to show \eqref{cond} or 
\eqq{
\Lie_e\eta^* = 0\qquad\iff\qquad \Lie_e\circ = 0,
}
where the equivalence follows from the definition \eqref{met} and the fact that $\Lie_e\tr = 0$.
To obtain $\Lie_e\circ$ one should use similar computation as in
the proof of next Lemma~\ref{lemmaE}.  
When $e=1$ the computation of $\Lie_e\circ$  is straightforward, for $e=1-\la_p$ it is
slightly more involved.  In the case of $e=1-\frac{1}{u_1}\la_p$, which is the exception, the relation $\Lie_e\circ = 0$
does not hold.
\end{proof}

\begin{lemma}\label{lemmaE} 
For the Frobenius algebras from Lemma~\ref{uvf} the quasi-homogeneity relation, $\Lie_E\circ = (d-1)\circ$, in the case of $\nu=\infty$ hold for the Euler vector field and the weight given by
\eq{\label{evf}
	E = \la - \frac{1}{n}p\la_p\qquad\text{and}\qquad d=1 + (s-1)\frac{2}{n}.
}
In the renaming cases the Euler vector field has the form $E=\la$ with $d=1$.
\end{lemma}
\begin{proof}
We will present the detailed proof only for the case of $\nu=\infty$. 
First notice that the vector field $E=\la - \frac{1}{n}p\la_p$ belongs to $\T_\la\Alg^\infty_M$. 
To compute $\Lie_E\circ$ we will use the formula~\eqref{liecirc}, that is
\eq{\label{liec}
\bra{\Lie_E\circ}(\alpha,\beta) = \bra{\Dir_E\circ}(\alpha,\beta)+ \Dir_{\alpha\circ\beta}^*E
	-\Dir_\alpha^*E\circ\beta - \alpha\circ\Dir_\beta^*E.
}
Let $X\in\vf{\Alg_M^\infty}$ and $\alpha\in\of{\Alg^\infty_M}$, then $\Dir_XE = X-\frac{1}{n}pX_p$
and 
\eqq{
	\dual{\Dir_\alpha^*E,X} = \dual{\alpha,\Dir_XE} = \tr_\infty\bra{\alpha\Dir_XE}
	= \tr_\infty\bra{\Big(\frac{n-s+1}{n}\alpha+\frac{1}{n}p\alpha_p\Big)X},
}
where the integration by parts is used in the trace form \eqref{traces} ($\nu=\infty$).
Hence, $\Dir_\alpha^*E = \frac{n-s+1}{n}\alpha+\frac{1}{n}p\,\alpha_p$. 
The subsequent terms in \eqref{liec} have the form: 
\eqq{
	\bra{\Dir_E\circ}(\alpha,\beta) &= \ell\bra{p^sE_p\alpha}\beta + \alpha \ell\bra{p^sE_p\beta}\\ 
	&= \frac{n-1}{n}\alpha\circ\beta - \frac{1}{n}\ell\bra{p^{s+1}\la_{2p}\alpha}_{\me s}^\infty\beta 
	    - \frac{1}{n}\alpha\ell\bra{p^{s+1}\la_{2p}\beta}_{\me s}^\infty,\\
	\Dir_{\alpha\circ\beta}^*E &= \frac{n-s+1}{n}\alpha\circ\beta + \frac{1}{n}p\, \bra{\alpha\circ\beta}_p,\\
	\Dir_\alpha^*E\circ\beta &= \ell\bra{p^s\la_p\Dir_\alpha^*E}\beta + \Dir_\alpha^*E\, \ell\bra{p^s\la_p\beta}\\ 
	&= \frac{n-s+1}{n}\alpha\circ\beta + \frac{1}{n}\ell\bra{p^{s+1}\la_p\alpha_p}_{\me s}^\infty\beta 
	    + \frac{1}{n}p\,\alpha_p\ell\bra{p^s\la_p\beta}_{\me s}^\infty.
}
Substituting the above terms into \eqref{liec} and using the fact that the relation $p\pr_p\ell(\cdot) = \ell(p\pr_p\cdot)$ holds for \eqref{ell} (for $s=0$ or $1$) we obtain the following equality
\eqq{
		\Lie_E\circ = 2\frac{s-1}{n}\circ.
}

For finite $\nu=0$ or $v$ we have $E=\la\in \T_\la\Alg^\nu_M$ and $D_XE = X$, $D_\alpha^*E=\alpha$. 
The corresponding computation of the quasi-homogeneity relation is similar and adequately simpler
then the above one. 
\end{proof}

Combining the above lemmas and propositions with Theorem~\ref{main} we have the following
result:

\begin{theorem}\label{thrm} 
\begin{itemize}
\item[]
\item For $s=0,1$ there is a structure of Frobenius manifold on $\Alg^\infty_M$ if 
$\deg_\infty\la\me 1$. 
\item There is a structure of Frobenius manifold on $\Alg^0_M$ for $s=0$ if $\deg_0\la=0$
and for $s=1$ if $\deg_0\la\me 1$. 
\item In the case of $s=1$ and $\deg_0\la = -1$
there is a structure of Frobenius manifold on $\Alg^0_M$ with nonflat unit vector field.
\item For finite $\nu=v$ and $s=0,1$ there is a structure of Frobenius manifold on $\Alg^v_M$ if $\deg_v\la \me -1$.
\end{itemize}
\end{theorem}

\begin{remark}
Consider the second contravariant metric on $\Alg^\nu_M$ defined by \eqref{gmet}
for $r=\ell^*$ and $E=\la$, which takes the form
\eq{\label{secmet}
	g^*(\alpha,\beta) = \tr_\nu\bra{\la\,\alpha\circ \beta}\qquad \alpha,\beta\in\of{\Alg_M^\nu}.
}
This metric is well defined and nondegenerate at a generic point of manifold subspaces $\Alg^\nu_M$
for finite $\nu=0$ or $v$ and $s=0,1$. This can be showed in a similar way to the proof of Proposition~\ref{ndm}.
In these cases, since $E=\la$ is a Euler vector field, the metric \eqref{secmet} coincides with the intersection form \eqref{if}. However, for $\nu=\infty$ the metric \eqref{secmet} is well defined and nondegenerate only on the  
spaces: 
 \eqq{
\text{for $s=0$}:\quad\tilde{\Alg}^\infty_M &= \pobr{\la(p) = p^n + \tilde{u}p^{n-1} + u_{n-2}p^{n-2} + \ldots},\\
\text{for $s=1$}:\quad\tilde{\Alg}^\infty_M &= \pobr{\la(p) = \tilde{u}p^n + u_{n-1}p^{n-1} + u_{n-2}p^{n-2} + \ldots}.
}
Hence, to obtain the second metric on $\Alg^\infty_M$ one must to carry out the reduction procedure with respect to the constrain $\tilde{u}=0$. After reduction one gets the metric in the form
$g^*_{\text{red}}(\alpha,\beta) = \dual{\alpha^\sharp,\beta}$,
where
\eqq{
	\alpha^\sharp = p^s\la\brac{\la_p \alpha}^\infty_{\me 0} - p^s\la_p \brac{\la \alpha}^\infty_{\me 0} 
+ \frac{1}{n} p^s \la_p\brac{\la_p \alpha}^\infty_{-1}\qquad s=0,1.
}
The above reduction is in fact equivalent to the so-called Dirac reduction of Lie-Poisson brackets considered in \cite{Szab}.
Using the relation $\brac{\,\cdot\,}^\infty_{-1} = \brac{p\,\cdot\,}^\infty_{\me 0} - p\brac{\,\cdot\,}^\infty_{\me 0}$
one finds that the metric $g^*_{\text{red}}$ coincides with the intersection form \eqref{if}, that is
\eqq{
g^*_{\text{red}}(\alpha,\beta) = \dual{E,\alpha\circ\beta} \equiv \tr_\infty\bra{E\,\alpha\circ\beta},
}
where $E= \la - \frac{1}{n}p\la_p$. The flatness of $g^*_{\text{red}}$ and compatibility with $\eta^*$ is a consequence
of the fact that  $g^*_{\text{red}}$ is an intersection form. Notice that $E= \la - \frac{1}{n}p\la_p$ does
not fulfil assumptions from Subsection~\ref{subs} and the theorem cannot be used with this choice of $E$. 
\end{remark}
\begin{remark}
The particular infinite-dimensional Frobenius manifold $\Alg_M^\infty$ for $s=0$ and $n=1$ corresponds to the
Frobenius manifold associated with dKP equation or Benney chain, which was constructed in \cite{Rai}.
On the other-hand in \cite{CDM}  there was formulated infinite-dimensional Frobenius manifold 
associated with two-component Toda chain, which corresponds to the following direct sum $\Alg^\infty_M\oplus\Alg^0_M$
for $s=1$ and $n=m=0$. 
\end{remark}

\subsection{Frobenius structure on the spaces of meromorphic functions} \label{frobsec}

Let us consider the algebra of meromorphic functions on the Riemann sphere with prescribed marked points, poles (or zeros), at $\infty$, $0$ and two families of finite (not fixed) points $a_1,\ldots a_L$ and $v_1, \ldots,v_K$ that can vary over
the complex plane. Therefore, we define the algebra $\Alg$ in the form\footnote{By $\Cm\brac{x,y,\ldots}$ we mean the ring of complex polynomials in variables $x, y,\ldots\,$.}
\eqq{
		\Alg = \Cm\brac{p,p^{-1},(p-a_1)^{-1},\ldots,(p-a_L)^{-1},(p-v_1)^{-1},\ldots,(p-v_K)^{-1}},
}
where $p\in\hat{\Cm}$. The elements $p,(p-a_i)^{-1},(p-v_j)^{-1}$ and the generators of the algebra $\Alg$ with obvious relations between them.

The underlying manifold subspaces $\Alg_M$ of $\Alg$
on which we are going to define the structure of Frobenius manifold are reductions of the infinite-dimensional Frobenius manifolds associated to the formal Laurent spaces considered in the previous section. Hence, we must take into consideration the constraints from Theorem~\ref{thrm} on the degrees of
the meromorphic functions $\la(p)\in\Alg_M$.  

Accordingly, from Theorem~\ref{thrm} it follows that for $s=0$ the meromorphic functions $\la(p)\in\Alg_M$
cannot have zero or pole at $p=0$, that is we must require $\deg_0\la =0$. For $s=1$ there must be singularity at $p=0$ of order $\deg_0\la\me 1$ or zero of order one, that is $\deg_0\la=-1$ .
All non fixed zeros $a_i$ of $\la(p)$ must be of order one, $\deg_{a_i}\la=1$, 
and all non fixed poles $v_j$ must be of degree $\deg_{v_j}\la\me 1$.   
Besides, the meromorphic functions $\la(p)$ must have singularity at $\infty$ of order $\deg_\infty\la\me 1$ and the normalisation from \eqref{ami} must be taken into account. 

Let admissible $n:=\deg_\infty\la$, $m_0:=\deg_0\la$ and $m_j:=\deg_{v_j}\la$ be fixed. Consequently, we define the underlying manifold subspace of $\Alg$ by
\eq{\label{am}
	\Alg_M = \pobr{\la(p) = \frac{\prod_{i=1}^{L}(p-a_i)}{p^{m_0}\prod_{j=1}^{K}(p-v_j)^{m_j}}
\ \bigg |\  a_i,v_i\in \Cm, \sum_{i=1}^La_i=\sum_{j=1}^K m_jv_j\text{ for $s=0$}},
}
where we require that $n = L-\sum_{j=0}^Km_j\me 1$ and $m_j\me 1$. Besides, for $s=0$ we must
have $m_0=0$, and for $s=1$ we must have $m_0\me 1$ or $m_0=-1$. 
The coefficients $a_i$ and $v_j$ constitute coordinates on the underlying manifold, which is of dimension 
\eqq{
N=K+L+s-1.
}

The related tangent spaces are spanned by the derivation of $\la$ with respect to the coordinates, that is 
\eqq{
	\T_\la\Alg_M = \spn\pobr{\frac{\pr\la}{\pr a_1},\ldots,\frac{\pr\la}{\pr a_L},\frac{\pr\la}{\pr v_1},\ldots,\frac{\pr\la}{\pr v_K}}.
}
Let $\Gamma :=\pobr{\infty, 0\,(\text{if $s=1$}),v_1,\ldots,v_K}$ be the set that consists of poles 
of meromorphic functions $\la\in\Alg_M$. Let's also define $\tilde{\Gamma}:=\Gamma\setminus \{\infty\}$.

Let's define the following subspaces of $\Alg$:
\eqq{
	\Alg^\infty_{\me k} := p^k\Cm\brac{p},\qquad
	\Alg^\infty_{< k} := \pobr{f\in\Alg\ |\ \deg_\infty f< k}
}
and
\eqq{
	\Alg^\nu_{\me k} = \pobr{f\in\Alg\ |\ \deg_\nu f\me -k},\qquad
	\Alg^\nu_{< k} = (p-\nu)^{-k-1}\Cm\brac{(p-\nu)^{-1}},
}
where $\nu\in\tilde{\Gamma}$. Let
\eqq{
factor : = p^{m_0}\prod_{j=1}^{K}(p-v_j)^{m_j+1}.
}
Then, the tangent space is given by\footnote{$\Cm_r\brac{p}$ is the space of complex polynomials in $p$ of degree at most $r$.}
\eqq{
\T_\la\Alg_M = factor^{-1}\times \Cm_{N-1}\brac{p}\equiv factor^{-1}\times \Cm_{N-1}\brac{p-\nu}, 
}
where $\nu\in\tilde{\Gamma}$. Notice that expanding into Laurent series
$\T_\la\Alg_M\subset \Alg^\infty_{\les n+s-2}$ and 
$\T_\la\Alg_M\subset\Alg^\nu_{\me -m-\varepsilon}$, where $\varepsilon=0$ for $s=1$ and $\nu=0$,
and otherwise $\varepsilon=1$. 

The respective orthogonal complements of the cotangent space are defined with respect to the trace forms \eqref{traces}. For  $\nu=\infty$ we have
\eq{\label{orth}
(\T_\la\Alg_M^\infty)^\perp = factor\times p^s\times\bra{\Alg^\infty_{\me 0}\oplus\Alg^\infty_{<-N}}
}
and for $\nu\in\tilde{\Gamma}$ we have
\eqq{
(\T_\la\Alg_M^\nu)^\perp = factor\times p^s\times\bra{\Alg^\nu_{\me 0}\oplus\Alg^\nu_{<-N}}.
}

The cotangent  spaces are defined by the respective quotient spaces, that is 
\eqq{
\T_\la^*\Alg_M^\nu\cong\Alg/(\T_\la\Alg_M^\nu)^\perp\qquad \nu\in\Gamma.
} 
As result the form of the cotangent spaces is not unique. Possible and convenient representations are:
\eq{\label{cotsp1}
	\T^*_\la\Alg_M^\infty = p^{1-n}\Cm_{N-1}\brac{p}\subset \Alg^\infty_{\me 1-n}
}
for $\nu=\infty$ and
\eq{\label{cotsp2}
	\T^*_\la\Alg_M^{\nu} = (p-\nu)^m\Cm_{N-1}\brac{(p-\nu)^{-1}}  \subset \Alg^{\nu}_{\les m} 
}
for $\nu\in\tilde{\Gamma}$, where $m=\deg_\nu\la$.

\begin{lemma}\label{froba}
Each metric \eqref{met}, defined for $\nu\in\Gamma$, is nondegenerate at a generic point $\la\in\Alg_M$. 
The orthogonal complements $(\T_\la\Alg_M^\nu)^\perp$ are ideals in $\Alg$ with respect to the multiplications
\eqref{mn} for respective $\nu\in\Gamma$. Moreover, in each case the metric \eqref{met} is compatible with the structure of quotient algebras $\Alg/(\T_\la\Alg_M^\nu)^\perp\cong \T_\la^*\Alg_M$.
\end{lemma}
\begin{proof}
We will consider in detail only the case of $\nu=\infty$. All the following computations are made at a generic point $\la\in\Alg_M$. The following analysis shall be done for $s=0$ and $s=1$ separately. Then, for arbitrary $\alpha\in\T^*_\la\Alg_M^\infty$ from \eqref{cotsp1} one finds that 
\eqq{
\begin{split}
factor\times p^s\brac{\la_p \alpha}^\infty_{\me 0}
&\in\Alg^\infty_{\me m_0+s} \\
factor\times p^s\la_p \brac{\alpha}^\infty_{\me 0}&\in\Alg^\infty_{\me 0}
\end{split}
\quad\Longrightarrow\quad \alpha^\sharp\in factor^{-1}\times\Alg^\infty_{\me 0},
}
where \eqref{sharp} is used. On the other hand 
\eqq{
\begin{split}
factor\times p^s\la_p \brac{\alpha}^\infty_{< 0} &\in\Alg^\infty_{< N}\\
factor\times p^s\brac{\la_p \alpha}^\infty_{< 0}&\in\Alg^\infty_{< N-n+1}
\end{split}
\quad\Longrightarrow\quad \alpha^\sharp\in factor^{-1}\times\Alg^\infty_{< N}.
}
Hence, one can conclude that $\im\sharp = factor^{-1}\times \Cm_{N-1}\brac{p}$, that is 
the image of $\sharp$ spans $\T_\la\Alg_M$. Besides, one can easily check that the image of the orthogonal complement $(\T_\la\Alg_M^\infty)^\perp$ with respect to the map \eqref{sharp} is empty.

The orthogonal complement to $\T_\la\Alg_M$ is given in the form \eqref{orth}, hence let
consider arbitrary $f\in factor\times p^s\times\Alg^\infty_{\me 0}$ and $g\in factor\times p^s\times\Alg^\infty_{< -N}$.
We see that $\la_p f\in\Alg^\infty_{\me 0}$ and $\la_p g\in\Alg^\infty_{< 0}$. Hence, 
\eqq{
	f\circ \Alg = p^s f\brac{\la_p \Alg}^\infty_{\me 0}= factor\times\Alg^\infty_{\me 2s}\subset (\T_\la\Alg_M^\infty)^\perp
}
and
\eqq{
	g\circ \Alg = - p^s g\brac{\la_p \Alg}^\infty_{< 0}= factor\times\Alg^\infty_{< 2s-N-1}\subset (\T_\la\Alg_M^\infty)^\perp,
}
where the multiplication is defined by \eqref{mn} with $\nu=\infty$. Thus $(\T_\la\Alg_M^\infty)^\perp$ is an ideal in $\Alg$
with respect to \eqref{mn}. 
	
Now, the compatibility of the quotient structure with the metric \eqref{met} follows
from the relations:
\eqq{
	\tr_\infty\bra{factor\times\Alg^\infty_{\me 2s}} = -\res_{p=\infty} \Alg^\infty_{\me m_0+s} = 0
}	
and
\eqq{
	\tr_\infty\bra{factor\times\Alg^\infty_{< 2s-N-1}} = -\res_{p=\infty} \Alg^\infty_{<-n-1} = 0.
}	

The remaining cases of $\nu\in\tilde{\Gamma}$ can be proven in a similar fashion or can be obtained as a corollary to the next theorem.
\end{proof}

For each $\Alg_M$ the duality paring, such as \eqref{dualp}, can be defined by the trace form \eqref{traces}
for different $\nu\in\Gamma$. As result, a $1$-form $\gamma$ on the underlying manifold can have different representations $\gamma_{\nu}\in\T^*_\la\Alg_M^\nu$, such that at point $\la\in\Alg_M$:
\eq{\label{eqv}
	\dual{\gamma,X}_\la \equiv \tr_\infty\bra{X\gamma_\infty} = \tr_{\nu}\bra{X\gamma_{\nu}}\qquad \nu\in\tilde{\Gamma},
}
where $X\in\T_\la\Alg_M$ is arbitrary. Then, for each $\nu\in\Gamma$ 
we can define  the related contravariant metric on $\Alg_M$  \eqref{met} and the related multiplication in the cotangent bundle $\T^*\Alg_M^\nu$ using  \eqref{mn}. We will show that these structures defined for different $\nu\in\Gamma$ are isomorphic.

Let $1$-form $\alpha$ be represented by $\alpha_{\nu}\in\T^*_\la\Alg_M^\nu$, then let $\alpha_\nu^\sharp\in\T_\la\Alg_M$ be given by \eqref{sharp} and let $\alpha_\nu\circ\beta_\nu$ means the product of two $1$-forms in the tangent bundle $\T^*\Alg_M^\nu$ defined by the multiplication \eqref{mn} for respective $\nu\in\Gamma$.

\begin{theorem}
The metrics defined  by \eqref{sharp} for different $\nu\in\Gamma$ on (fixed) $\Alg_M$ are
equivalent, that is the following relation is true:
\eqq{
     \eta^*\bra{\alpha,\beta} := \tr_\infty\bra{\alpha_\infty^\sharp\beta_\infty} \equiv \tr_{\nu}\bra{\alpha_{\nu}^\sharp\beta_{\nu}}\qquad \nu\in\tilde{\Gamma}. 
}
This means that for arbitrary $1$-form $\alpha$ the following equality is also valid
\eq{\label{eqv2}
\alpha^\sharp := \alpha^\sharp_\infty\equiv \alpha^\sharp_{\nu}\in \T_\la\Alg_M\qquad \nu\in\tilde{\Gamma}.
}
Similarly, the multiplications in the cotangent bundles $\T^*\Alg^\nu_M$ defined by \eqref{mn}
for different $\nu\in\Gamma$ are isomorphic. 
\end{theorem}
\begin{proof}
Let $\nu\in\tilde{\Gamma}$. Using  \eqref{eqv} and the residue theorem one finds that
\eqq{
 &\tr_\nu\bra{\alpha_\nu^\sharp\beta_\nu} = \tr_\infty\bra{\alpha_\nu^\sharp\beta_\infty}
 = -\res_{p=\infty}\bra{\brac{\la_p \alpha_\nu}^\nu_{\me 0}\beta_\infty - \la_p \brac{\alpha_\nu}^\nu_{\me 0}\beta_\infty}\\
&\qquad = -\res_{p=\infty}\bra{\brac{\la_p \alpha_\nu}^\nu_{\me 0}\brac{\beta_\infty}^\infty_{<0} -  \brac{\alpha_\nu}^\nu_{\me 0}\brac{\la_p\beta_\infty}^\infty_{< 0}}\\
&\qquad = \res_{p=\nu}\bra{\brac{\la_p \alpha_\nu}^\nu_{\me 0}\brac{\beta_\infty}^\infty_{<0} - \brac{\alpha_\nu}^\nu_{\me 0}\brac{\la_p\beta_\infty}^\infty_{< 0}}\\
&\qquad = \res_{p=\nu}\bra{\la_p \alpha_\nu\brac{\beta_\infty}^\infty_{<0} - \alpha_\nu\brac{\la_p\beta_\infty}^\infty_{< 0}}
= \tr_\nu\bra{\beta_\infty^\sharp\alpha_\nu} = \tr_\infty\bra{\beta_\infty^\sharp\alpha_\infty}.
}
Hence, the equivalence of metrics \eqref{met} is proven.

The multiplications in the cotangent bundles $\T^*\Alg^\nu_M$ defined by \eqref{mn} are mutually isomorphic 
if for arbitrary $X\in\T_\la\Alg_M$ the following relation holds:
\eqq{
	\dual{X,\beta\circ\gamma} = \tr_\infty\bra{X\bra{\beta_\infty\circ\gamma_\infty}}
	\equiv \tr_\nu\bra{X\bra{\beta_\nu\circ\gamma_\nu}}.
}
Using \eqref{eqv2} we can  take $X=\alpha^\sharp_\infty \equiv \alpha^\sharp_\nu$, where the $1$-form $\alpha$ is arbitrary. Hence, the above relation  is equivalent to
\eqq{
	\tr_\infty\bra{\alpha_\infty\circ\beta_\infty\circ\gamma_\infty} = \tr_\nu\bra{\alpha_\nu\circ\beta_\nu\circ\gamma_\nu}, 
}
where $\alpha,\beta,\gamma$ are arbitrary $1$-forms. By \eqref{arel} for $\nu\in\Gamma$:
\eqq{
	(\alpha_\nu\circ\beta_\nu)^\sharp &= p^s\alpha^\sharp\brac{\la_p\beta_\nu}^\nu_{\me 0} -p^s\la_p\brac{\alpha^\sharp\beta_\nu}^\nu_{\me 0}\\ &= p^s\la_p\brac{\alpha^\sharp\beta_\nu}^\nu_{< 0} - p^s\alpha^\sharp\brac{\la_p\beta_\nu}^\nu_{< 0}.
}
Hence, in the same manner as before, using  \eqref{eqv} and the residue theorem, one finds that
\eqq{
&\tr_\nu\bra{\alpha_\nu\circ\beta_\nu\circ\gamma_\nu} = \tr_\nu\bra{\alpha_\nu\bra{\beta_\nu\circ\gamma_\nu}^\sharp}
 = \tr_\infty\bra{\alpha_\infty\bra{\beta_\nu\circ\gamma_\nu}^\sharp}\\
 &\qquad= -\res_{p=\infty}\bra{\alpha_\infty\beta^\sharp\brac{\la_p\gamma_\nu}^\nu_{\me 0} - \la_p\alpha_\infty\brac{\beta^\sharp\gamma_\nu}^\nu_{\me 0}}\\
 &\qquad= \res_{p=\nu}\bra{\brac{\alpha_\infty\beta^\sharp}^\infty_{<0}\la_p\gamma_\nu - \brac{\la_p\alpha_\infty}^\infty_{<0}\beta^\sharp\gamma_\nu}\\
 &= \tr_\nu\bra{\bra{\alpha_\infty\circ\beta_\infty}^\sharp\gamma_\nu} 
 = \tr_\infty\bra{\bra{\alpha_\infty\circ\beta_\infty}^\sharp\gamma_\infty} = \tr_\infty\bra{\alpha_\infty\circ\beta_\infty\circ\gamma_\infty}.
}
This finishes the proof.
\end{proof}

\begin{lemma}\label{leme}
On each underlying manifold $\Alg_M$ there is a structure of Frobenius algebra defined 
by the quotient algebra  $\Alg/(\T_\la\Alg^\nu_M)^\perp\cong \T_\la^*\Alg^\nu_M$  and the respective metric \eqref{met}.
These structures are equivalent for different $\nu\in\Gamma$. The related unit vector fields are given by
\eqq{
e = 
\begin{cases}
1-\la_p & \text{for $\nu=\infty$, $s=0$ and $\deg_\infty\la = 1$,}\\
1-\frac{1}{u_1}\la_p &  \text{for $\nu=0$, $s=1$ and $\deg_0\la = -1$,}\\
1 & otherwise.
\end{cases}
} 
So, almost always $e=1$. In all the above cases but one the respective unit vector fields are flat. The exception is the case of $s=1$ with $\deg_0\la= -1$.

For the above Frobenius algebras the quasi-homogeneity relation, $\Lie_E\circ = (d-1)\circ$ is fulfilled by the Euler vector field $E = \la - \frac{1}{n}p\la_p$ and the weight $d=1 + (s-1)\frac{2}{n}$. 
\end{lemma}

The first part of Lemma~\eqref{leme} is a corollary to previous propositions and the proof of the second part is practically the same as the respective parts of the proofs of Lemma~\ref{uvf} and Lemma~\ref{lemmaE}.

\begin{proposition}\label{flatcoord}
The flat coordinates for the contravariant metric \eqref{met} defined on $\Alg_M$ are given by
\eqq{
t_\infty^i := \frac{1}{1-\frac{i}{n}}\tr_\infty \la^{1-\frac{i}{n}}\quad\text{for}\quad 1\les i\les n-1 
}
and for $1-s\les k \les K$ by \footnote{Care must be taken when calculating the traces of the terms involving logarithmic singularities, see Appendix~\ref{bcint}.}
\eq{\label{ls1}
t_{v_k}^j =
\begin{cases}
  \frac{1}{1-\frac{j}{m_k}}\tr_{v_k}\la^{1-\frac{j}{m_k}}    & \text{for}\quad  0\les j< m_k, \\
   \tr_{v_k} \log\la + \frac{m_k}{n}\tr_{\infty} \log\la    & \text{for}\quad j=m_k.
\end{cases}
}
The respective differentials  are
\eqq{
	dt_\infty^i := \brac{\la^{-\frac{i}{n}}}^\infty_{\me 1-n}\in \T^*_\la\Alg_M^\infty\qquad 1\les i\les n-1
}
and
\eqq{
	dt_{v_k}^j := \brac{\la^{-\frac{j}{m_k}}}^{v_k}_{\les m_k}\in \T^*_\la\Alg_M^{v_k}\qquad
1-s\les k \les K,\quad 0\les j\les m_k.
}
\end{proposition}
\begin{proof}
First we will show that $dt_\nu^i$ are flat $1$-forms
with respect to the metric defined by \eqref{met} for $\nu\in\Gamma$.   
See Remark~\ref{rem}, it is sufficient to check if the equalities \eqref{flatc} (with $E=1$) hold.
Remind that $r=\ell^*$, where $\ell^*$ is given by \eqref{elld}.

For $\nu=\infty$ we take $\gamma^i_\infty = \la^{-\frac{i}{n}}$, for which  
$\deg_\infty\gamma^i_\infty = -i$. Hence, the conditions \eqref{flatc}
are satisfied since $\brac{\gamma^i_\infty}^\infty_{\me 0} = 0$ for $i>0$. 
Projecting $\gamma^i_\infty$ on the cotangent space \eqref{cotsp1}, one finds
that the only nonzero projections are for $1\les i\les n-1$.  Similarly for $\nu\in\tilde{\Gamma}$ 
we take $\gamma^j_\nu = \la^{-\frac{i}{m}}$, where $m=\deg_\nu\la$. Thus $\deg_\nu\gamma^j_\nu = -j$ 
and \eqref{flatc} hold as $\brac{\gamma^j_\nu}^\nu_{< 0} = 0$ for $j\me 0$.
The only nonzero projections of $\gamma^j_\nu$ on the respective cotangent space \eqref{cotsp2} 
are for $0\les j\les m$.  

The related flat coordinates are the respective (locally defined) functions $t^i_\nu$ so that
\eqq{
	\gamma^i_\nu =  dt^i_\nu\qquad\iff\qquad \Dir_X t^i_\nu = \tr_\nu\bra{X \gamma^i_\nu},
}
where $X$ is an arbitrary vector field on $\Alg_M$.
\end{proof}

\begin{proposition}
The contravariant metric $\eta^*$ on $\Alg_M$, defined by means of \eqref{met}, 
decomposes in flat coordinates from Proposition~\ref{flatcoord}  into anti-diagonal
blocks such that
\eqq{
	\eta^*\bra{dt^i_\nu,dt^j_\nu} = m\, \delta_{i,m-j}\qquad m=\deg_\nu\la\qquad \nu\in\Gamma.
} 
\end{proposition}
\begin{proof} First notice that $\bra{dt^i_\infty}^\sharp = p^s\brac{\la_p dt^i_\infty}^\infty_{\me 0}$
and $\bra{dt^j_{\nu}}^\sharp = -p^s\brac{\la_p dt^j_{\nu}}^{\nu}_{< 0}$ for $\nu\in\tilde{\Gamma}$.
Then, we have
\eqq{
	\eta^*\bra{dt^i_\infty, dt^j_{\nu}} &= \tr_\infty\bra{dt^i_\infty \bra{dt^j_{\nu}}^\sharp}
	= -\res_{p=\infty}\bra{\la^{-\frac{i}{n}}\brac{\la_p \la^{-\frac{j}{m}}}^{\nu}_{< 0}}\\
	&= -\res_{p=\infty}\bra{\brac{\la^{-\frac{i}{n}}}^\infty_{\me 0}\brac{\la_p \la^{-\frac{j}{m}}}^{\nu}_{< 0}} = 0
}
and for different $\nu, \nu'\in\tilde{\Gamma}$:
\eqq{
	\eta^*\bra{dt^i_{\nu}, dt^j_{\nu'}} &= \tr_{\nu}\bra{dt^i_{\nu}\bra{dt^j_{\nu'}}^\sharp}
	= \res_{p=\nu}\bra{\la^{-\frac{i}{m}}\brac{\la_p \la^{-\frac{j}{m'}}}^{\nu'}_{< 0}}\\
	&= \res_{p=\nu}\bra{\brac{\la^{-\frac{i}{m}}}^{\nu}_{< 0}\brac{\la_p \la^{-\frac{j}{m'}}}^{\nu'}_{< 0}} = 0.
}
Besides, 
\eqq{
	\eta^*\bra{dt^i_\infty, dt^j_\infty} &= \tr_\infty\bra{dt^i_\infty \bra{dt^j_\infty}^\sharp}
	= -\res_{p=\infty}\bra{\la^{-\frac{i}{n}}\brac{\la_p \la^{-\frac{j}{n}}}^\infty_{\me 0}}\\
	&= -\res_{p=\infty}\bra{\la_p \la^{-\frac{i+j}{n}}} = -n \res_{\la=\infty}\bra{\la^{-\frac{i+j}{n}}}
	= n\, \delta_{i,n-j}
}
where in the residue integral one makes the change of coordinates $p\map \la = p^n + \ldots$, taking into account the multiplicity of this transformation, $n=\deg_\infty\la$. Similarly, for $\nu\in\tilde{\Gamma}$:
\eqq{
	\eta^*\bra{dt^i_{\nu}, dt^j_{\nu}} &= \tr_{\nu}\bra{dt^i_{\nu} \bra{dt^j_{\nu}}^\sharp}
	= \res_{p=\nu}\bra{\la^{-\frac{i}{m}}\brac{\la_p \la^{-\frac{j}{m}}}^{\nu}_{< 0}}\\
	&= \res_{p=\nu}\bra{\la_p \la^{-\frac{i+j}{m}}} = -m \res_{\la=\infty}\bra{\la^{-\frac{i+j}{m}}}
	= m\, \delta_{i,m-j},
}
where the change of coordinates is of the form $p\map \la = \ldots + (p-\nu)^{-m} $ and $m=\deg_\nu\la$. 
\end{proof}

\begin{proposition}\label{eulervf}
The coefficients of the Euler vector field $E = \la - \frac{1}{n}p\la_p$ in the flat coordinates, that is
$E^i_\nu\equiv  \dual{dt^i_\nu, E}$, for $\nu=\infty$ are  
\eqq{
	E^{i}_\infty = \bra{\frac{1-s}{n}+\frac{n-i}{n}}t^i_{\infty}\quad\text{for}\quad 1\les i\les n-1 
}
and for $\nu\in\tilde{\Gamma}$ are
\eqq{
E_{\nu}^j =
\begin{cases}
 \bra{\frac{1-s}{n}+\frac{m-j}{m}}t^j_{\nu}
      & \text{for}\quad  0\les j< m, \\
   \frac{1}{n} t^m_{\nu}   & \text{for}\quad s=0,\, j=m\\
   \frac{m_0}{n}+1 & \text{for}\quad s=1,\, \nu=0,\, j=m\\
   \frac{m}{n} & \text{for}\quad s=1,\, \nu\neq0,\, j=m ,
\end{cases}
}
where $m=\deg_\nu\la$.
\end{proposition}
\begin{proof}
For $\nu=\infty$ the derivation is as follows  
\eqq{
E^{i}_\infty &= \Der_E t^{i}_{\infty}  = \tr_\infty\bra{\la^{-\frac{i}{n}} E} 
= -\res_{p=\infty}\bra{p^{-s}\la^{-\frac{i}{n}}\bra{ \la - \frac{1}{n}p\la_p}}\\
&= -\bra{1+\frac{1-s}{n-i}}\res_{p=\infty}\la^{1-\frac{i}{n}} = \bra{\frac{1-s}{n}+\frac{n-i}{n}}t^i_\infty.
 }
In particular for $s=1$, $\nu=0$ and $j=m_0$ we have
\eqq{
	\dual{dt^{m_0}_{0},E} &= \Der_E t^{m_0}_{0} 
	= \tr_{0}\bra{\la^{-1}E} + \frac{m_0}{n}\tr_{\infty}\bra{\la^{-1}E}\\ 
	&= \res_{p=0}\bra{p^{-1}-\frac{1}{n}\la^{-1}\la_p} 
	- \frac{m_0}{n}\res_{p=\infty}\bra{p^{-1}-\frac{1}{n}\la^{-1}\la_p}\\
&= \res_{p=0}\bra{p^{-1}+\frac{m_0}{n}p^{-1}} 
	- \frac{m_0}{n}\res_{p=\infty}\bra{p^{-1}-p^{-1}} = \frac{m_0}{n} + 1.
}
The computation in all other cases is similar.
\end{proof}

In order to use Lemma~\ref{Fpro} we need the following proposition, proof of which is straightforward.

\begin{proposition}\label{hden}
Define, by means of the flat coordinates from Proposition~\ref{flatcoord}, functions 
$\H^{\nu,i}_{(0)} := t^i_\nu$. The recurrence formula \eqref{recH} takes the form
\eqq{
	\frac{\pr \H^{\nu,i}_{(n)}}{\pr \la} = \H^{\nu,i}_{(n-1)}.
}
Hence,  we have
\eqq{
\mathcal{H}^{\infty,i}_{(1)} := \frac{1}{1-\frac{i}{n}}\frac{1}{2-\frac{i}{n}}\tr_\infty \la^{2-\frac{i}{n}}\quad\text{for}\quad 1\les i\les n-1 
}
and for $1-s\les k \les K$ we have 
\eq{\label{ls2}
\H^{v_k,j}_{(1)} =
\begin{cases}
  \frac{1}{1-\frac{j}{m_k}}\frac{1}{2-\frac{j}{m_k}}\tr_{v_k}\la^{2-\frac{j}{m_k}}    & \text{for}\quad  0\les j< m_k, \\
   \tr_{v_k} \bra{\la\log\la-\la} + \frac{m_k}{n}\tr_{\infty}\bra{\la\log\la-\la}    & \text{for}\quad j=m_k.
\end{cases}
}
\end{proposition}

Now, combining the above results and using Lemma~\ref{Fpro}
we have the following theorem.

\begin{theorem}
In all cases but one there is a structure of Frobenius manifold on $\Alg_M$. 
The exception is the case of $s=1$ with $\deg_0\la = -1$, when
there is a structure of Frobenius manifold with a nonflat unit vector field.
All the ingredients of these structures are defined above.
The respective prepotential functions have the following form 
\eq{\label{prepot}
\F = \frac{1}{3-d}\bra{\frac{1}{n}\sum_{i=1}^{n-1}E^i_{\infty} \H^{\infty,n-i}_{(1)}
+ \sum_{k=1-s}^K\frac{1}{m_k}\sum_{j=0}^{m_k}E^j_{v_k}\H^{v_k,m_k-j}_{(1)}},
}
where $d=1+(s-1)\frac{2}{n}$.
\end{theorem}

\begin{remark}
The class of Frobenius manifolds constructed in this section corresponds
to the Frobenius manifolds classified in \cite{Dub1} and associated with 
Hurwitz spaces of zero genus. With respect to the related Landau-Ginzburg formalism,  the cases of $s=0$ and $s=1$ correspond to the choice of primary differentials as $d\omega = dp$ and $d\omega = \frac{dp}{p}$, respectively.   
The particular case of $\Alg_M$ for $s=0$, which consists of polynomial functions \eqref{polf},
is associated with Frobenius manifolds arising in the Saito's (singularity) theory labeled by $A_{n-1}$, see \cite{Dub1}.
On the other-hand,  $\Alg_M$ for $s=1$, which consists of meromorphic functions with poles only at infinity and zero,
 is related with a class of Frobenius manifolds studied in \cite{Dub4} (see also \cite{Car})
associated with the extended affine Weyl groups of the $A$ series. 
The explicit form of the prepotential \eqref{prepot} is a close analog of respective formulae
given in \cite{Dub1} and \cite{K2} (for $s=0$).    
\end{remark}

Now, we will illustrate the presented theory with few characteristic examples, some of them will contain more details 
of the related computations than other. The scheme is as follows. First, one needs to establish
the manifold subspace $\Alg_M$ \eqref{am} and compute the flat coordinates $t_1,\ldots,t_N$
according to Proposition~\ref{flatcoord}. The flat coordinates can be chosen so that
the unit $e=\frac{\pr}{\pr {t_1}}$ (except the case of $s=1$ with $m_0 = \deg_0\la=-1$).
Next, the prepotential function $\F$ is given by \eqref{prepot}, where the formulae 
from Propositions~\ref{eulervf} and \ref{hden} must be used. The respective Euler vector fields can be obtained from Lemma~\ref{leme} or Proposition~\ref{eulervf}. Having the prepotential $\F$ coefficients of the covariant metric $\eta$ and 
the structure constants of the multiplication in the tangent bundle can be easily computed from \eqref{c}.

\begin{example}
Consider the manifold space $\Alg_M$, in the case of $s=0$, that consists of meromorphic
functions with only one pole at infinity of fixed order $n$: 
\eq{\label{polf}
	\la = p^n + u_{n-2}p^{n-2} + \ldots + u_1 p + u_0\qquad n\me 2.
}
In this case the prepotential \eqref{prepot} takes the following 'symmetric' form: 
\eqq{
	\F = \frac{n^2}{2(n+1)} \sum_{i=1}^{n-1} \frac{n+1-i}{i(n^2-i^2)}\res_{p=\infty}\la^{\frac{n-i}{n}} \res_{p=\infty}\la^{\frac{n+i}{n}}.
}

In particular case of $n=4$ we have 
\eqq{
\la  = p^4 + u p^2 + v p + w \equiv p^4 + t_3 p^2 + t_2 p + t_1 + \frac{1}{8} t_3^2,
}
where the flat coordinates are given by
\eqq{
t_1\equiv t^{1}_\infty &= -\frac{4}{3}\res_{p=\infty}\la^{\frac{3}{4}} = w-\frac{1}{8}u^2\\
t_2\equiv t^{2}_\infty &= -2\res_{p=\infty}\la^{\frac{1}{2}} = v\\
t_3\equiv t^{3}_\infty &= -4\res_{p=\infty}\la^{\frac{1}{4}} = u.
}
Using the above formula one obtains
\eqq{
\F = \frac{1}{8} t_1 t_2^2+\frac{1}{8} t_1^2 t_3-\frac{1}{64} t_2^2 t_3^2+\frac{t_3^5}{3840}.
}
The related Euler vector field is
\eqq{
	E = t_1\frac{\pr}{\pr {t_1}} + \frac{3}{4}t_2\frac{\pr}{\pr {t_2}}
+\frac{1}{2}t_3\frac{\pr}{\pr {t_3}}\quad\iff\quad E(\la) =\la - \frac{1}{4}p\la_p
}
and the weight $d=\frac{1}{2}$. The unit vector field is $e=\frac{\pr}{\pr {t_1}}$, since $e(\la)=1$.
\end{example}

\begin{example}
Consider $\Alg_M$, for $s=0$, that consists of functions with pole at infinity and 
two non-fixed 'finite' poles $v_1=v$ and $v_2=w$, all of order one, that is $n=m_1=m_2=1$. 
Hence, $\la\in\Alg_M$ has the form
\eqq{
\la = p + \frac{a}{p- v} + \frac{b}{p- w} = 
 p + \frac{t_3}{p- t_1 - t_2} + \frac{t_4}{p- t_1 + t_2}.
}
We take the following flat coordinates:
\eqq{
t_1 &\equiv \frac{1}{2}\bra{t_v^1 + t_w^1} = \frac{1}{2}\bra{v+w}\\
t_2 &\equiv \frac{1}{2}\bra{t_v^1 - t_w^1} = \frac{1}{2}\bra{v-w}\\
t_3 &\equiv t_v^0 = a\\
t_4 &\equiv t_w^0 = b,
}
where, for $\nu=v,w$, by Proposition~\ref{flatcoord}:
\eqq{
t^0_\nu = \res_{p=\nu}\la,\qquad t^1_\nu = \res_{p=\nu}\log \la -\res_{p=\infty}\log \la.
}
By Proposition~\ref{hden}:
\begin{align*}
\H^{\nu,0}_{(1)} &= \frac{1}{2}\res_{p=\nu}\la^2\\
 \H^{\nu,1}_{(1)} &=  \res_{p=\nu}\bra{\la \log\la - \la}- \res_{p=\infty}\bra{\la \log\la - \la},
\end{align*}
hence using the formula \eqref{prepot} the prepotential is
\begin{align*}
\F = t_1 t_2 (t_3-t_4)+\frac{1}{4} (t_1^2 + t_2^2 ) (t_3+t_4) + \frac{1}{2}t_3^2 \log t_3 
+ t_3t_4 \log t_2 +\frac{1}{2}t_4^2 \log t_4.
\end{align*}
The related Euler vector field 
\eqq{
	E = t_1\frac{\pr}{\pr t_1} + t_2\frac{\pr}{\pr t_2}
+2t_3\frac{\pr}{\pr t_3} + 2t_4\frac{\pr}{\pr t_4} \quad\iff\quad E(\la) =\la - p\la_p,
}
weight $d=-1$ and the unit vector field $e=\frac{\pr}{\pr {t_1}}$ as $e(\la)=1-\la_p$.
\end{example}

\begin{example}
Consider $\Alg_M$, for $s=1$, consisting of functions with poles at infinity, $0$ and 
non-fixed 'finite' pole $v_1=w$, all of order one, that is $n=m_0=m_1=1$: 
\eqq{
\la = p + u+\frac{v}{p} + \frac{s}{p- w} = 
p + t_1+t_3+\frac{e^{t_2}}{p} - \frac{t_3\, e^{t_4}}{p + e^{t_4}}.
}
The flat coordinates are:
\eqq{
t_1 \equiv t_0^0 = u-\frac{s}{w},\qquad 
t_2 \equiv t_0^1 = \log{v},\qquad
t_3 \equiv t_w^0 = \frac{s}{w},\qquad
t_4 \equiv t_w^1 = \log (-w).
}
One obtains the prepotential
\begin{align*}
\F = \frac{1}{2} t_1^2 t_2+t_1 t_3 t_4+\frac{1}{2} t_3^2 t_4+e^{t_2}+t_3e^{t_2-t_4} -t_3e^{t_4} +\frac{1}{2} t_3^2 \log t_3. 
\end{align*}
Besides,
\eqq{
	E = t_1\frac{\pr}{\pr t_1} + 2\frac{\pr}{\pr t_2}
+t_3\frac{\pr}{\pr t_3} + \frac{\pr}{\pr t_4} \quad\iff\quad E(\la) =\la - p\la_p,
}
$d=1$ and $e=\frac{\pr}{\pr {t_1}}$ as $e(\la)=1$.
\end{example}

\begin{example}
Consider $\Alg_M$, for $s=1$, consisting of functions with poles of order one at infinity and $0$, that is
$n=m_0=1$: 
\eqq{
\la = p + u+\frac{v}{p}  = p + t_1+\frac{e^{t_2}}{p},
}
where the flat coordinates are $t_1 \equiv t_0^0 = u$ and $t_2 \equiv t_0^1 = \log{v}$.
The potential is 
\eqq{
\F = \frac{1}{2} t_1^2 t_2+e^{t_2}.
}
Besides, 
\eqq{
	E = t_1\frac{\pr}{\pr t_1} + 2\frac{\pr}{\pr t_2} \quad\iff\quad E(\la) =\la - p\la_p,
}
$d=1$ and $e=\frac{\pr}{\pr {t_1}}$ as $e(\la)=1$. This is a celebrated example of Frobenius manifold
corresponding to the quantum cohomology of complex projective line $\mathbb{P}^1$. 
\end{example}

\begin{example}
Let $s=1$ and $\Alg_M$ consists of functions with poles of order one at infinity and $v_1=v$, that is
$n=m_1=1$ as well as $m_0=\deg_0\la=-1$ : 
\eqq{
\la = p + u + \frac{uw}{p- w} = 
p + t_1 - \frac{t_1\, e^{t_2}}{p + e^{t_2}},
}
with flat coordinates $t_1 \equiv t_w^0 = u$ and $t_2 \equiv t_w^1 = \log (-w)$. The prepotential is
\begin{align*}
\F = \frac{1}{2} t_1^2 t_2  -t_1e^{t_2} +\frac{1}{2} t_1^2 \log t_1. 
\end{align*}
Besides,
\eqq{
	E = t_1\frac{\pr}{\pr t_1}+ \frac{\pr}{\pr t_2} \quad\iff\quad E(\la) =\la - p\la_p,
}
$d=1$ and $e=\frac{1}{t_1+e^{t_2}}\bra{t_1\frac{\pr}{\pr t_1}- \frac{\pr}{\pr t_2}}$, since $e(\la)=1+\frac{w}{u-w}\la_p$. 
This is example when the unit field $e$ is not flat. This particular example of Frobenius manifold was considered very recently in \cite{BCR}.
\end{example}

\begin{example}
The case of $s=0$ and $A_M$ consisting of meromorphic function with pole of order one
at infinity ($n=1$) and $v\equiv t_1$ ($m_1=1$):
\eqq{
	\la = p+\frac{t_3}{p-t_1}+\frac{t_2^2}{(p-t_1)^2},
}
where the flat coordinates are: $t_1\equiv \frac{1}{2}t_v^2$, $t_2\equiv \frac{1}{2}t_v^1$ and $t_3\equiv t_v^0$. 
One obtains the prepotential
\eqq{
	\F =  t_1 t_2^2 + \frac{1}{2} t_1^2 t_3 + \frac{1}{2} t_3^2 \log t_2,
}
with
\eqq{
	E = t_1\frac{\pr}{\pr t_1} + \frac{3}{2}t_2\frac{\pr}{\pr t_2}  + 2t_3 \frac{\pr}{\pr t_3} \quad\iff\quad E(\la) =\la - p\la_p,
}
$d=-1$ and $e=\frac{\pr}{\pr {t_1}}$ as $e(\la)=1-\la_p$.
\end{example}

\begin{example}
The case of $s=1$, and six-dimensional Frobenius manifold associated with $\Alg_M$
for $n=m_0=2$ and $m_1=1$. Thus, 
\eqq{
	\la = p^2 +  t_4\, p + t_1 + t_5 + \frac{t_2\, e^{\frac{t_3}{2}}}{p}+ \frac{e^{t_3}}{p^2}  - \frac{t_5\,e^{t_6}}{p+e^{t_6}},
}
where the flat coordinates are $t_1\equiv t_0^0$, $t_2\equiv t_0^1$, $t_3\equiv t_0^2$,
$t_4\equiv t_\infty^1$, $t_5\equiv t_v^0$ and $t_6\equiv t_v^1$ ($v\equiv -e^{t_6}$).
One obtains
\eqq{
	\F &= -\frac{t_2^4}{96}+\frac{1}{4} t_1 t_2^2 
	- \frac{t_4^4}{96} + \frac{1}{4} t_1 t_4^2 + \frac{1}{4}t_1^2 t_3   +\frac{1}{4} t_4^2 t_5+\frac{1}{2} t_5^2 t_6+t_1 t_5t_6 
     + t_2t_4 e^{\frac{t_3}{2}} \\
    &\quad\    +  \frac{e^{t_3}}{2} + \frac{1}{2}t_5 e^{2 t_6}  - t_4 t_5e^{t_6} 
    + t_2 t_5e^{\frac{t_3}{2}-t_6} - \frac{1}{2} t_5e^{t_3-2 t_6}    + \frac{1}{2} t_5^2 \log t_5
}
and
\eqq{
	E = t_1\frac{\pr}{\pr t_1} + \frac{1}{2}t_2\frac{\pr}{\pr t_2} + 2 \frac{\pr}{\pr t_3}
	+\frac{1}{2}t_4 \frac{\pr}{\pr t_4} + t_5\frac{\pr}{\pr t_5}+ \frac{1}{2} \frac{\pr}{\pr t_6} \quad\iff\quad 
	E(\la) =\la - \frac{1}{2}p\la_p,
}
$d=1$ and $e=\frac{\pr}{\pr {t_1}}$ as $e(\la)=1$.	
\end{example}

%\newpage

\smallskip

\appendix

\section{Convention and notation}\label{A}

Here, we fix convention and notation as well as provide useful or required in the main text facts from the  rather
standard differential geometry, all in a coordinate-free form.

\paragraph{\bf Directional derivative}\label{a1} 

Let $M$ be a smooth manifold. The directional derivative (G\^ateaux derivative) at a point $p\in M$ 
of a tensor field $T$ is\footnote{For some chart $x:U\subset M\arrow \Km^n$ the formula for the directional derivative
takes the form
\eqq{
		\bra{\Dir_X T}(x) = \Diff{T\bra{x+\varepsilon X(x)}}{\varepsilon}{\varepsilon=0}, 
}
which is particularly useful in computations.}
\eqq{
	\bra{\Dir_X T}(p) = \Diff{T\bra{\phi_\varepsilon (p)}}{\varepsilon}{\varepsilon=0}\qquad p\in M, 	
}
where $X\in\vf{M}$ is some vector field and $\phi_\varepsilon$ is its flow, i.e. $X(p)=\Diff{\phi_\varepsilon (p)}{\varepsilon}{\varepsilon=0}$ and $\phi_0(p) = p$. Then, 
\eq{\label{dxf}
\Dir_X f \equiv X(f) = \dual{df,X},
}
where $\dual{\cdot,\cdot}:\of{M}\times\vf{M}\arrow \smf{M}$ is the ordinary duality pairing,
$f\in\smf{M}$ is a smooth function and $df\in\of{M}$ is its differential.
The Lie bracket between two vector fields $X$ and $Y$ takes the form 
\eq{\label{lbr}
\brac{X,Y} = \Dir_XY - \Dir_Y X.
} 
Using the above definition one can show that 
\eq{\label{xy}
		\Dir_{\brac{X,Y}}T = \Dir_X\Dir_Y T - \Dir_Y\Dir_X T.  
}
In fact, the directional derivative is the simplest example of a symmetric (torsionless) and flat covariant derivative.

\paragraph{\bf Lie derivative} 

Let $X\in\vf{M}$. The Lie derivative of functions coincides with the directional derivative \eqref{dxf}, $\Lie_X f = \Dir_X f$. The Lie derivative of vector fields coincides with the Lie bracket \eqref{lbr}, $\Lie_X Y = \brac{X,Y}$.
Through the Leibniz rule its action can be uniquely extended on any tensor field.\footnote{Alternatively
the Lie derivative of a tensor field $T$ can be defined as 
\eqq{
	(\Lie_X T)(p) = \Diff{(\phi_\varepsilon^*T)(p)}{\varepsilon}{\varepsilon=0}\qquad p\in M,	
}
where $\phi_\varepsilon^*$ is the usual pull-back of the flow $\phi_\varepsilon$ of $X$.} 

\begin{proposition} Let $\gamma\in\of{M}$. Then,
\eq{\label{lieof}
	\Lie_X\gamma = \Dir_X\gamma + \Dir_\gamma^*X,
}
where $\dual{\Dir_\gamma^*X,Y} := \dual{\gamma,\Dir_YX}$. 
The Lie derivative of a tensor field $\circ:\of{M}\otimes\of{M}\arrow \of{M}$, 
$\circ(\alpha,\beta)\equiv \alpha\circ\beta$,
is
\eq{\label{liecirc}
	\bra{\Lie_X\circ}(\alpha,\beta) = \bra{\Dir_X\circ}(\alpha,\beta)+ \Dir_{\alpha\circ\beta}^*X
	-\Dir_\alpha^*X\circ\beta - \alpha\circ\Dir_\beta^*X.
}
\end{proposition}
\begin{proof}
From the Leibniz rule we have
\eqq{
	\dual{\Lie_X\gamma,Y} = \Lie_X\dual{\gamma,Y} - \dual{\gamma,\Lie_X Y} 
	=  \dual{\Dir_X\gamma,Y} + \dual{\gamma,\Dir_Y X}.
}
Hence, \eqref{lieof} follows. In a similar fashion, using \eqref{lieof} and
\eqq{
	\dual{(\Lie_X\circ)(\alpha,\beta),Y} &= \Lie_X\dual{\alpha\circ\beta,Y}  - \dual{\alpha\circ\beta,\Lie_XY}\\
	&\qquad\qquad- \dual{\Lie_X\alpha\circ\beta,Y} - \dual{\alpha\circ\Lie_X\beta,Y}
}
one obtains \eqref{liecirc}.
\end{proof}

\paragraph{\bf Levi-Civita connection}

Let the manifold $M$ be equipped with a (pseudo-Riemannian) covariant metric 
$\eta\in\Gamma\bra{S^2\T^*M}$. Then, $\eta$ and its inverse (contravariant metric) 
$\eta^*\in\Gamma\bra{S^2\T M}$ induces canonical isomorphisms
\eqq{
\begin{split}
&\sharp:\of{M}\arrow\vf{M}\qquad  \alpha\map \alpha^\sharp := \sharp(\alpha),\\
&\flat:\vf{M}\arrow\of{M}\qquad  X\map X^\flat:= \flat(X) 
\end{split}
}
such that $\flat\circ\sharp = \sharp\circ\flat = \id$ and
\eqq{
\eta^*(\alpha,\beta) = \dual{\alpha,\beta^\sharp} = \eta (\alpha^\sharp, \beta^\sharp).
}

The (unique) Levi-Civita connection is a covariant derivative 
$\nabla:\vf{M}\times\vf{M}\arrow\vf{M}$ fulfilling the following two requirements:
\begin{itemize}
\item[i)] $\nabla$ is symmetric (torsionless), 
\eq{\label{symm}
\nabla_XY-\nabla_YX = \brac{X,Y};
}
\item[ii)] $\nabla$ preserves the metric, that is $\nabla \eta = 0$. Equivalently
\eq{\label{comp}
X\bra{\eta\bra{Y,Z}} = \eta\bra{\nabla_XY,Z} + \eta\bra{Y,\nabla_XZ}.
}
\end{itemize}
Let $\gamma\in\of{M}$, then \eqref{symm} is equivalent to
\eqq{
d\gamma(X,Y) = \dual{\nabla_X\gamma,Y} - \dual{\nabla_Y\gamma,X}.
}
Since $d\gamma(X,Y) = \dual{\Dir_X\gamma,Y} - \dual{\Dir_Y\gamma,X}$,
the first condition for $\nabla$ becomes
\eqq{
	\dual{\nabla_X\gamma - \Dir_X\gamma,Y} = \dual{\nabla_Y\gamma - \Dir_Y\gamma,X}.
}
The second condition can be rewritten in the form
\eqq{
\bra{\Dir_X\eta^*}\bra{\beta,\gamma} = \eta^*\bra{\nabla_X\beta-\Dir_X\beta,\gamma} + \eta^*\bra{\beta,\nabla_X\gamma-\Dir_X\gamma},
}
where $\beta,\gamma\in\of{M}$.

Using the canonical isomorphism between tangent and cotangent bundles,
induced by the metric $\eta$, we can define 
\eq{\label{gamma}
	\Gamma_\alpha\gamma\equiv \Gamma(\alpha,\beta) := \Dir_{\alpha^\sharp}\gamma - \nabla_{\alpha^\sharp}\gamma
	\qquad \alpha,\gamma\in\of{M}.
}
It is a tensor field since $\Gamma$ is obviously $\smf{M}$-bilinear map. In the coordinates in which the directional and covariant derivatives are taken $\Gamma$ coincides withe the Christoffel symbols of the Levi-Civita connection.
Hence, the following lemma is valid.

\begin{lemma}\label{cc} 
A covariant derivative is the Levi-Civita connection iff the following conditions are satisfied:
\eq{\label{c1}
		\eta^*\bra{\alpha,\Gamma_\beta\gamma} = \eta^*\bra{\beta,\Gamma_\alpha\gamma}  
}
and
\eq{\label{c2}
\bra{\Dir_{\alpha^\sharp}\eta^*}(\beta,\gamma)  + \eta^*\bra{\Gamma_\alpha\beta,\gamma}
+ \eta^*\bra{\beta,\Gamma_\alpha\gamma} = 0 .
}
\end{lemma}

To calculate the curvature tensor of the Levi-Civita connection we find
\eqq{
\nabla_{\alpha^\sharp}\nabla_{\beta^\sharp}\gamma = \Dir_{\alpha^\sharp}\Dir_{\beta^\sharp}\gamma
-\bra{\Dir_{\alpha^\sharp}\Gamma}(\beta,\gamma)
- \Gamma_{\Dir_{\alpha^\sharp}\beta}\gamma -  \Gamma_\beta\Dir_{\alpha^\sharp}\gamma
- \Gamma_\alpha\Dir_{\beta^\sharp}\gamma
+ \Gamma_\alpha\Gamma_\beta\gamma
}
and
\eqq{
\nabla_{[\alpha^\sharp,\beta^\sharp]}\gamma = \Dir_{[\alpha^\sharp,\beta^\sharp]}\gamma
-\Gamma_{[\alpha^\sharp,\beta^\sharp]^\flat}\gamma.
}
Since the connection is torsionless, we have
\eq{\label{lieb}
[\alpha^\sharp,\beta^\sharp]^\flat = \nabla_{\alpha^\sharp}\beta - \nabla_{\beta^\sharp}\alpha
= \Dir_{\alpha^\sharp}\beta -\Dir_{\beta^\sharp}\alpha - \Gamma_\alpha\beta + \Gamma_\beta\alpha.
}
Hence, using the relation \eqref{xy} we derive the curvature tensor in the form
\eq{
\begin{split}\label{ct}
&R(\alpha^\sharp,\beta^\sharp)\gamma = \nabla_{\alpha^\sharp}\nabla_{\beta^\sharp}\gamma - \nabla_{\beta^\sharp}\nabla_{\alpha^\sharp}\gamma - \nabla_{[\alpha^\sharp,\beta^\sharp]}\gamma\\
&\qquad = \bra{\Dir_{\beta^\sharp}\Gamma}(\alpha,\gamma) - \bra{\Dir_{\alpha^\sharp}\Gamma}(\beta,\gamma) +
\Gamma_\alpha\Gamma_\beta\gamma - \Gamma_\beta\Gamma_\alpha\gamma 
- \Gamma_{\Gamma_\alpha\beta}\gamma + \Gamma_{\Gamma_\beta\alpha}\gamma.
\end{split}
}

\section{Hydrodynamic Poisson bracket}\label{hydro}

\paragraph{\bf Poisson structures of hydrodynamic type}

Let $\eta$ be a flat (pseudo-Riemannian) covariant metric on a manifold $M$.
Then, on the loop space $\mathcal{L} (M) \equiv \smf{\Si,M}$, that consists of smooth maps of a circle
to $M$, one can define Poisson bracket of the hydrodynamic type:
\begin{equation}\label{pb}
\pobr{H, F} = \int_{\Si} \pd{f}{u^i}\pi^{ij}\pd{h}{u^j}\,dx := 
\int_{\Si} \ddual{df, \bra{\nabla_{u_x}dh}^\sharp} dx,
\end{equation}
where 
\eqq{
H=\int_{\Si} h(u(x))\, dx,\qquad F=\int_{\Si} f(u(x))\, dx
} 
are functionals on $\mathcal{L} (M)$, $u:\Si\arrow M$, $x\map u(x) = (u^1(x),\ldots,u^n(x))$
and $u^i(x)$ are local coordinate fields,  $\nabla_{u_x}$ is the Levi-Civita connection of $\eta$ taken along the vector field $u_x \equiv \diff{u}{x}$ tangent to a loop. The related Poisson tensor in this  local coordinates 
has the form
\begin{equation}\label{hf}
\pi^{ij} = \eta^{ij}\diff{}{x} - \eta^{ik}\Gamma_{kl}^{j}u^{l}_{x},
\end{equation}
where $\diff{}{x}\equiv \diff{u^k}{x}\frac{\pr}{\pr u^k}$ is the total derivative with respect to $x\in\Si$.  Important is fact that the invariance with respect to change of coordinates is preserved on the level of the infinite-dimensional formalism of the hydrodynamic Poisson brackets. Recall that \eqref{pb} is the well-known Dubrovin-Novikov bracket \cite{DN}. Field bracket defined, through an operator of the form \eqref{hf}, by means of a nondegenerate  matrix $\eta^{ij}$ is a Poisson bracket iff $\eta_{ij}$ can be interpreted as a flat covariant metric and $\Gamma_{kl}^{j}$ as the corresponding Levi-Civita connection.

The Poisson brackets of the form \eqref{pb} naturally arise in the study of Hamiltonian structures for $(1+1)$-dimensional hydrodynamic (dispersionless) systems. This class of systems is described by quasi-homogeneous first order PDE's.
To any flat pencil of contravariant metrics one can associate integrable hydrodynamic hierarchies with 
bi-Hamiltonian structure defined by means of hydrodynamic Poisson brackets. As it turns out,
under some homogeneity assumptions, natural geometric setting for the formalism of hydrodynamic bi-Hamiltonian structures is the theory of Frobenius manifolds \cite{Dub1,Dub2}.

\paragraph{\bf Principle hierarchy}

The functions $\H^k_{(n)}$ from \eqref{rec2} taken as
densities of corresponding functionals on the loop manifold associated to the Frobenius manifold generate hierarchies
of Hamiltonian hydrodynamic systems called as a Principal hierarchy \cite{Dub1}. In flat coordinates of the metric $\eta$
these hierarchies take the form
\eq{\label{princh}
\frac{\pr t^i}{\pr T^{(n)}_k} = \eta^{ij}\frac{\pr^2 \H^k_{(n)}}{\pr x\pr t^j} 
= \eta^{ij}\frac{\pr^2 \H^k_{(n)}}{\pr t^j\pr t^k}\frac{\pr t^k}{\pr x}\qquad n\me 0,
}
where $t^i = t^i(x)$ are dynamical fields depending on the infinite sets of evolution parameters (times) $T^{(n)}_k$. The recurrence formula \eqref{rec2} is a counterpart of the bi-Hamiltonian recursion scheme
and in fact the hierarchies \eqref{princh} can be written in a quasi bi-Hamiltonian form
with respect to hydrodynamic Poisson brackets generated  by the metric $\eta$ and the intersection form $g$.

\section{Classical $r$-matrix formalism on Poisson algebras}\label{crm}

Let $\bra{\alg,\brac{\cdot,\cdot}}$ be a Lie algebra. Classical $r$-matrix \cite{Sem1}
is a linear map $r:\alg \arrow \alg$ such that
\begin{equation*}
  \brac{a,b}_r := \brac{r(a), b} + \brac{a, r(b)}\qquad a,b\in\alg,
\end{equation*}
defines second Lie bracket on $\alg$. Fulfilling of the modified Yang-Baxter equation: 
\begin{equation*}
  \brac{r(a), r(b)} - r\bra{\brac{a, b}_r} + \kappa\, \brac{a,b} = 0\qquad \kappa\in\Km,
\end{equation*}
is a sufficient condition for a linear map $r$ 
to be a classical $r$-matrix.  Simplest solutions can be obtained through appropriate decomposition of $\alg$ into Lie subalgebras, that is $\alg = \alg_+\oplus\alg_-$.
Then, $r = \frac{1}{2}\bra{P_+-P_-}$, where $P_+,P_-$ are respective projections onto respective Lie subalgebras, is a classical $r$-matrix, since it satisfies the Yang-Baxter equation for $\kappa =\frac{1}{4}$.

\begin{theorem}[\cite{Li}]\label{liep}
 Let  $\bra{\aalg,\{\cdot,\cdot\},\cdot}$ have a structure of  Poisson algebra, that is $\aalg$ is unital, commutative algebra and the Lie bracket $\{\cdot,\cdot\}$ is a derivation with respect to the multiplication. 
Assume also existence of a non-degenerate $\ad$-invariant scalar product $(a,b)_\aalg\equiv \tr(a b)$ on $\aalg$, which means that $\aalg^*\cong\aalg$ and $\ad^*\cong\ad$. Then,  if $r:\aalg \arrow \aalg$
is a classical $r$-matrix, the formula
\eq{\label{pbn}
\pobr{h,f}_n(\lambda) = \bra{\lambda, \pobr{r(\lambda^ndh),df}+\pobr{dh,r(\lambda^ndf)}}_\aalg\qquad
h,f\in\mathcal{F}(\aalg),
}
defines for each $n\me 0$ Poisson bracket. Moreover, all these brackets are mutually compatible. 
\end{theorem}

 The related Poisson tensors $\pi_n$, such that $\pobr{h,f}_n = \bra{df, \pi_n dh}_\aalg$, have the following form
\eq{\label{ptn}
\pi_n dh = \pobr{\la,r(\la^ndh)}+\la^nr^*\bra{\pobr{\la,dh}},
}
where the adjoint of $r$ is defined by $(r^*(a),b)_\aalg:=(a,r(b))_\aalg$.  
For $n=1,2$ and $3$ the above Poisson brackets and related Poisson operators are respectively referred as linear, quadratic and cubic Poisson structures.

\section{Integrals involving logarithmic singularities}\label{bcint}

Here, we explain how to deal with the residuum integrals from \eqref{ls1} and \eqref{ls2} involving logarithmic singularities, that is integrals with branch points \cite{AF}. We want to compute the following integral for $\mu\me 0$:
\eqq{
I := \tr_{v_k}\bra{\la^\mu\log\la} + \frac{m_k}{n}\tr_{\infty}\bra{\la^\mu\log\la}\qquad s=0,1,
}
where $\la$ is meromorphic function with poles at $\infty$ and $v_i$ for $1-s\les i\les K$ ($v_0\equiv 0$ if $s=1$).
Recall that $\tr$ is given by \eqref{traces}. In the first integral we write $\log\la = \log\brac{(p-v_k)^{m_k}\la} - m_k\log (p-v_k)$ and in the second
integral: $\log\la = \log\brac{(p-v_k)^{-n}\la} + n\log (p-v_k)$. Hence,
\eqq{
I = \tr_{v_k}\bra{\la^\mu\log\brac{(p-v_k)^{m_k}\la}} + \frac{m_k}{n}\tr_{\infty}\bra{\la^\mu\log\brac{(p-v_k)^{-n}\la}}  
+ m_k\, I_0,
}
where
\eqq{
	I_0 = -\tr_{v_k}\bra{\la^\mu\log (p-v_k)} + \tr_{\infty}\bra{\la^\mu\log (p-v_k)}.
}
One can now write these two integrals in $I_0$ as one with contour surrounding a branch cut  
between $v_k$ and $\infty$. Now, after applying the residue theorem we have    
\eqq{
I_0 = \sum_{\substack{1-s\les j\les K\\ j\neq k}}\tr_{v_j}\bra{\la^\mu\log (p-v_k)}.
}

\smallskip

\section*{Acknowledgement}
The significant part of this research was supported by the European Community
under a Marie Curie Intra-European Fellowship, contract no.
PIEF-GA-2008-221624. The author wishes to express his thanks to Ian~Strachan for many fruitful and valuable
discussions and to the Department of Mathematics of Glasgow University for the hospitality. 
The author also wishes to thank M.~Semenov-Tian-Shansky for the information that the relation \eqref{RB}
is the Rota-Baxter identity.

%\footnotesize


\begin{thebibliography}{99}
%\itemsep-1mm
%\parsep-1mm

\bibitem{AF}  Ablowitz M. J. and Fokas A. S., {\it Complex variables: introduction and applications}, second edition, Cambridge University Press, Cambridge, 2003

\bibitem{AK} Aoyama S. and Kodama Y., {\it Topological conformal
field theory with a rational $W$ potential and the dispersionless
KP hierarchy}, Modern Phys. Lett. A {\bf 9} (1994) 2481-2492 

\bibitem{AK2} Aoyama S. and Kodama Y., {\it Topological Landau-Ginzburg theory with
a rational potential and the dispersionless KP hierarchy}, Comm.
Math. Phys. {\bf 182} (1996) 185-219 

\bibitem{Aud} Audin, M., {\it Symplectic geometry in Frobenius manifolds and quantum cohomology}, J. Geom. Phys. {\bf 25} (1998) 183-204


\bibitem{Bax} Baxter G., {\it An analytic problem whose solution follows from a simple algebraic identity}, Pacific J. Math. {\bf 10} (1960) 731-42

\bibitem{BS} B\l aszak M. and Szablikowski B. M., {\it Classical $R$-matrix theory for bi-Hamiltonian field systems},
J. Phys. A: Math. Theor. {\bf 42} (2009) 404002

\bibitem{BCR} Brini A., Carlet G. and  Rossi P., {\it Integrable hierarchies and the mirror model of local $\mathbb{CP}^1$}, Physica D {\bf 241} (2012) 2156-2167

\bibitem{Car} Carlet G., {\it The extended bigraded Toda hierarchy}, J. Phys. A {\bf 39} (2006) 9411-9435

\bibitem{CDM} Carlet G., Dubrovin B. and Mertens L. P., {\it Infinite-dimensional Frobenius manifolds for $2+1$ integrable systems}, Math. Ann. {\bf 349} (2011) 75-115

\bibitem{DVV} Dijkgraaf R., Verlinde H. and Verlinde E. {\it Topological strings in $d<1$},
Nuclear Phys. B {\bf 352} (1991) 59-86

\bibitem{Din1} Dinar Y. I., {\it On classification and construction of algebraic Frobenius manifolds}, J. Geom. Phys. {\bf 58} (2008) 1171-1185

\bibitem{Din2} Dinar Y. I., {\it Yassir Ibrahim Frobenius manifolds from regular classical W-algebras}, Adv. Math. {\bf 226} (2011) 5018-5040

\bibitem{Dub0} Dubrovin B., {\it Integrable systems in topological field theory}, Nuclear Phys. B {\bf 379} (1992) 627-689

\bibitem{Dub1} Dubrovin B., {\it Geometry of $2$D topological field theories}, in "Integrable systems and quantum groups" (Montecatini Terme, 1993),  120-348, Lecture Notes in Math., 1620, Springer, Berlin, 1996


\bibitem{Dub2} Dubrovin B., {\it Flat pencils of metrics and Frobenious manifolds}, Proceedings of Taniguchi Symposium "Integrable systems and algebraic geometry" (Kobe/Kyoto, 1997) 47--72, World Sci. Publ., 1998

\bibitem{Dub4} Dubrovin B. and Zhang Y., {\it Extended affine Weyl groups and Frobenius manifolds},
Compositio Math. {\bf 111} (1998) 167-219

\bibitem{Dub3} Dubrovin B., Liu S.-Q. and Zhang Y., {\it Frobenius manifolds and central invariants for the Drinfeld-Sokolov biHamiltonian structures}, Adv. Math. {\bf 219} (2008) 780-837

\bibitem{DN} Dubrovin B. and Novikov S., {\it The Hamiltonian formalism of one-dimensional systems of hydrodynamic
type and the Bogolyubow--Witham averaging method}, Sov. Math. Dokl. {\bf 27} (1983) 665

\bibitem{Gue} Guest M. A., {\it From quantum cohomology to integrable systems}, Oxford Graduate Texts in Mathematics, Oxford University Press, Oxford, 2008

\bibitem{Guo} Guo L., What is? a Rota-Baxter algebra?; {\it Notices Amer. Math. Soc.} {\bf 56} (2009) 1436-37

\bibitem{Her} Hertling C., {\it Frobenius manifolds and moduli spaces for singularities}, Cambridge Tracts in Mathematics 151, Cambridge University Press, Cambridge, 2002

\bibitem{Hit} Hitchin N., {\it Frobenius manifolds}, with notes by David Calderbank, in "Gauge theory and symplectic geometry" (Montreal, 1995), 69-112, Kluwer Acad. Publ., Dordrecht, 1997


\bibitem{KM} Konopelchenko B. G. and Magri F. {\it Coisotropic deformations of associative algebras and dispersionless integrable hierarchies}, Comm. Math. Phys. {\bf 274} (2007) 627-658

\bibitem{Kon} Konopelchenko B. G. {\it Quantum deformations of associative algebras and integrable systems}, J. Phys. A {\bf 42} (2009) 095201

\bibitem{K2} Krichever I. M., {\it The $\tau$-function of the universal Whitham hierarchy, matrix models and topological field theories}, Comm. Pure Appl. Math. {\bf 47} (1994) 437-475

\bibitem{Li} Li L. C., {\it Classical r-Matrices and Compatible Poisson Structures for Lax Equations in Poisson Algebras}, Comm. Math. Phys. {\bf 203} (1999) 573-592

\bibitem{Man} Manin Y. I., {\it Frobenius manifolds, quantum cohomology, and moduli spaces},
American Mathematical Society Colloquium Publications 47, AMS, Providence, 1999

\bibitem{Rai} Raimondo A., {\it Frobenius manifold for the dispersionless Kadomtsev-Petviashvili equation}, Comm. Math. Phys. {\bf 311} (2012) 557-594

\bibitem{Rot} Rota G.-C., Baxter algebras and combinatorial identities. I, II; {\it Bull. Amer. Math. Soc.} {\bf 75}
(1969) 325-29

\bibitem{Sem1} Semenov-Tian-Shansky M. A., {\it What is a classical r-matrix?}, Funct. Anal. Appl. {\bf 17} (1983) 259

\bibitem{Sem2} Semenov-Tian-Shansky M. A., {\it Integrable systems and factorization problems}, Factorization and integrable systems (Faro, 2000), 155-218, Oper. Theory Adv. Appl. {\bf 141} Birkh�user, Basel, 2003

\bibitem{Str1} Strachan I. A. B, {\it Degenerate Frobenius manifolds and the bi-Hamiltonian structure of rational Lax equations}, J. Math. Phys. {\bf 40}  (1999) 5058-5079

\bibitem{Str2} Strachan I. A. B., {\it Frobenius manifolds: natural submanifolds and induced bi-Hamiltonian structures}, Differ. Geom. Appl. {\bf 20} (2004) 67-99

\bibitem{Szab} Szablikowski B. M. and B\l aszak M., {\it Meromorphic Lax representations of $(1+1)$-dimensional multi-Hamiltonian dispersionless systems}, J. Math. Phys. {\bf 47} (2006) 092701

\bibitem{Witt} Witten E., {\it On the structure of the topological phase of two-dimensional gravity},
Nuclear Phys. B {\bf 340} (1990) 281-332

\bibitem{WX} Wu C.-Z. and Xu D., {\it A class of infinite-dimensional Frobenius manifolds and their submanifolds}, Int. Math. Res. Not. {\bf 19} (2012) 4520-4562

\bibitem{WZ} Wu C.-Z. and Zuo D., {\it Infinite-dimensional Frobenius manifolds underlying the Toda lattice hierarchy},  Adv. Math. {\bf 255} (2014)  487-524

\end{thebibliography}
\end{document}